\documentclass{sig-alternate}

\usepackage[ruled,vlined,linesnumbered]{algorithm2e}
\usepackage{multirow}
\usepackage{scalefnt}
\usepackage{amsmath}
\usepackage{amssymb}
\usepackage{graphics}
\usepackage{epsfig}
\usepackage[usenames,dvipsnames]{xcolor}

\newcommand{\argmax}{\operatornamewithlimits{argmax}}

\newdef{definition}{Definition}
\newtheorem{theorem}{Theorem}
\newdef{example}{Example}

\newcommand\xqed[1]{%
  \leavevmode\unskip\penalty9999 \hbox{}\nobreak\hfill
  \quad\hbox{#1}}
\newcommand\exend{\xqed{$\blacksquare$}}

\usepackage{amsmath}

\usepackage{pgfplots}
\usepackage{graphicx}
\usepackage{subfigure}

\usepackage{tikz}
\usetikzlibrary{matrix}
\usetikzlibrary{arrows}
\usetikzlibrary{positioning}
\usetikzlibrary{patterns}

\pgfdeclarepatternformonly{my grid}{\pgfqpoint{-2pt}{-2pt}}{\pgfqpoint{4pt}{4pt}}{\pgfqpoint{1.25pt}{1.25pt}}%
{
  \pgfsetlinewidth{0.4pt}
  \pgfpathmoveto{\pgfqpoint{0pt}{0pt}}
  \pgfpathlineto{\pgfqpoint{0pt}{3.1pt}}
  \pgfpathmoveto{\pgfqpoint{0pt}{0pt}}
  \pgfpathlineto{\pgfqpoint{3.1pt}{0pt}}
  \pgfusepath{stroke}
}

\usetikzlibrary{shapes.geometric}
\tikzset{
    >=stealth',
    shorten >=1pt,
    auto,
    thick,
    main node/.style={circle,draw,font=\sffamily\scriptsize\bfseries},
    simple node/.style={circle,draw,font=\sffamily\small},
    rectangle node/.style={rectangle,thick,font=\sffamily\small, draw=black, minimum height=0.5cm, minimum width=1.5cm},
    node distance=1.4cm,
    root node/.style={diamond, draw,font=\sffamily\large\bfseries},
    chain node/.style={rectangle, draw,font=\sffamily\large\bfseries},    
}

\usepackage{hyperref}
\hypersetup{
colorlinks=true, 
linkcolor=black, 
citecolor=black, 
filecolor=black, 
urlcolor=black 
}

\begin{document}
\setlength{\textfloatsep}{5pt}
\setlength{\floatsep}{0pt}
\setlength{\intextsep}{1\baselineskip}

\title{An Effective Marketing Strategy for Revenue Maximization with a Quantity Constraint}

\author{%
{Ya-Wen Teng{\small $~^{\#1}$},  Chih-Hua Tai{\small $~^{*2}$}, Philip S. Yu{\small $~^{\dag3}$}, Ming-Syan Chen{\small $~^{\#4}$} }%
\vspace{1mm}\\
\fontsize{10}{10}\selectfont\itshape
$^{\#}$\,EE Dept., National Taiwan University, Taipei, Taiwan\\
\fontsize{9}{9}\selectfont\ttfamily\upshape
%
$^{1}$\,ywteng@arbor.ee.ntu.edu.tw, 
$^{4}$\,mschen@cc.ee.ntu.edu.tw%
\vspace{0mm}\\
\fontsize{10}{10}\selectfont\rmfamily\itshape
$^{*}$\,CSIE Dept., National Taipei University, New Taipei, Taiwan\\
\fontsize{9}{9}\selectfont\ttfamily\upshape
$^{2}$\,hanatai@mail.ntpu.edu.tw
\vspace{0mm}\\
\fontsize{10}{10}\selectfont\rmfamily\itshape
$^{\dag}$\,CS Dept., University of Illinois at Chicago, Chicago, IL 60607\\
\fontsize{9}{9}\selectfont\ttfamily\upshape
$^{3}$\,psyu@uic.edu
}

\maketitle
\begin{abstract}
Recently the influence maximization problem has received much attention for its applications on viral marketing and product promotions.
However, such influence maximization problems have not taken into account the 
monetary effect on the purchasing decision of
individuals. To fulfill this gap, in this paper, we aim for
maximizing the revenue by considering the quantity constraint on the promoted commodity. For
this problem, we not only identify a proper small group of
individuals as seeds for promotion but also determine the pricing of the commodity.
To tackle the revenue maximization problem, we first introduce a
strategic searching algorithm, referred to as Algorithm PRUB, which
is able to derive the optimal solutions. After that, we further
modify PRUB to propose a heuristic, Algorithm PRUB+IF, for obtaining feasible solutions more efficiently on larger instances. 
Experiments on real social networks with different valuation distributions
demonstrate the effectiveness of PRUB and PRUB+IF. 

\end{abstract}

\vspace{-2mm}
\section{Introduction}
\label{Sec:Intro}

Due to the advance of the Web 2.0 techniques,
various kinds of websites have become parts of our life. 
Most people nowadays are used to sharing, seeking and obtaining information, 
and interacting with others through various social networking websites such as Digg, Slashdot, Facebook, and Twitter, to name a few.
As more and more people join these social websites, 
the phenomenon of immense quantity of
information flow and influence spread becomes prominent. 
This then motivates the utilization of this phenomenon, known as the ``word-of-mouth''
effect, to bring significant potential benefits in many types of business such as viral marketing and innovation promotion. 
For example, a company can effectively promote its
brand and new products by involving a group of influential people spreading the good words over social websites.  
Hence, in very recent
years, there is active research \cite{Chen2010, Domingos2001, Kempe2003, Singer2012} exploring the problem of influence maximization. 

Existing works studied the influence maximization problem generally based on two kinds of diffusion models, 
the Linear Threshold (LT) and the Independent Cascade (IC) models. 
In the LT model, the influence strength between a pair of individuals is modeled as a fixed weight, 
and the influence on an individual is described as cumulative. 
An individual is activated when the sum of the weights propagated from his/her activated neighbors exceeds his/her own threshold. 
The IC model assumes that each influence of one to another is an independent chance to activate the latter with a probability. 
On both models, the problem of influence maximization aims at identifying a proper group (of a specified size) 
of individuals as seeds to have the number of activated people maximized.  
However, note that such influence maximization problems have not taken into account the monetary aspect, which usually plays a crucial factor in people's purchase decisions.
That is, maximizing the influence spread does not necessarily bring maximum revenue, as 
an individual may be very interested in a new smartphone, 
but does not purchase the smartphone due to its high price.

According to the studies in the fields of social psychology \cite{Asch1951, Easley2010} and economics \cite{Fromlet2012, Katz1985, Sundararajan2007}, one's valuation towards a commodity will be positively influenced by those that have purchased the commodity, which is known as the herd mentality or positive network externalities. In other words, people tend to imitate others in their consumer behaviors, which will result in irrational valuations of the commodity, i.e., higher valuations beyond their inherent valuations.
Based on these phenomena, Mirrokni et al.\ \cite{Mirrokni2012} addressed the \textbf{Revenue Maximization} problem by incorporating the monetary concept into the maximization problem of influence spread based on the LT model. In the monetary LT model, 
the influence strength (or said weight) is transformed into the valuation concept, and the individual threshold is regarded as the inherent valuation.
One's valuation of a commodity is thus the sum of his/her inherent valuation and the valuation increment from social influences.
An individual then pays for a commodity and propagates his/her influence whenever her/his valuation is larger than or equal to the pricing of the commodity.   
Consider the network shown in Figure \ref{Fig:ex} as an example, where the number contained in each node is the inherent valuation and the number on an edge from one node to another is the weight, i.e., the influence strength. 
For the ease of illustration, without loss of generality, the function $F(x)=x$ is used to transform the effects of the social influences cumulated on an individual into the valuation concept in this example. If $f$ purchases the commodity, $c$'s valuation is increased to $\$3+F(1)=\$4$. If $a$, $e$, and $f$ all purchase the commodity, the valuation increment for $c$ is $F(3+2+1)=\$6$ and $c$'s valuation thus becomes $\$3+\$6=\$9$.
The objective of the revenue maximization problem thus aims at determining the pricing of the commodity and identifying a seed group of individuals for promotion so that the total revenue is maximized.

However, to the best of our knowledge, none of the existing works 
has taken a quantity constraint on commodities into account. 
Note that due to the consideration of marketing strategies or practical conditions, quantities of commodities are often constrained in the real world. 
For example, the scarcity fallacy usually makes commodities more desirable \cite{Worchel1975}, as a company may want to release limited edition commodities to boost customer desire.
Another example is the case of promoting a concert, where the quantity of the concert tickets is constrained by the number of available seats. In this case, due to the schedules of the musicians or singers, it is hard to add more shows and the total number of concert tickets. 
Hence, selecting seeds to receive free commodities will reduce the quantities of commodities that can be sold.
This makes the prior works unapplicable to determine the proper pricing and obtain maximum revenue. 
For example, suppose that the network shown in Figure~\ref{Fig:ex} is considered, where the diffusion model with the monetary concept in the work \cite{Mirrokni2012} is adopted and the function $F(x)=x$ is used to transform the effects of the social influence into the valuation concept.
Without consideration of the commodity quantity, the maximum revenue is $\$28$ obtained at the pricing of $\$7$ with the seed group $\{d,f\}$. 
With the influence from $\{d,f\}$, $a$ and $e$ are the first to purchase the commodity, since both their valuations $\$2+F(5)=\$7$ are equal to the pricing of $\$7$. Then, $b$ and $c$ will subsequently purchase the commodity due to their valuations $\$0+F(2+4+4)=\$10$ and $\$3+F(3+2+1)=\$9$, respectively.
However, if the quantity of commodities is limited to 4, 
the revenue will decrease to $\$14$ since two free commodities given to the seeds $d$ and $f$ should be taken off first. 
In this case, letting the commodity be priced at $\$6$ and choosing only $d$ as the seed will be the best. 
Under the influence of $d$, $a$ first purchases the commodity since its valuation $\$2+F(5)=\$7$ is larger than the pricing of $\$6$. Then, both $b$ and $c$ purchase the commodities due to their valuations $\$0+F(2+4)=\$6$ and $\$3+F(3)=\$6$, respectively. Finally, since $a$, $b$, and $c$ all pay for the commodity sold at the price of $\$6$, the maximum revenue is $\$18$.
Therefore, when the quantity of commodities is limited, 
we argue the need of new approaches for the revenue maximization problem. 

\begin{figure}[t]
\centering
\begin{tikzpicture}
    \node[main node] (a) {$a(\$2)$};
    \node[main node] (c) [right=2cm of a] {$c(\$3)$};
    \node[main node] (b) [above of=c] {$b(\$0)$};
    \node[main node] (d) [below of=c] {$d(\$1)$};
    \node[main node] (e) [right=2cm of b] {$e(\$2)$};
    \node[main node] (f) [below of=e] {$f(\$0)$};

    \path
        (a) edge [->, bend left=20] node {2} (b)
            edge [->, bend right=20] node [swap] {3} (c)
        (b) edge [->] node {1} (a)
            edge [->, bend left=20] node [swap] {2} (f)
        (c) edge [->] node {3} (d)
        (d) edge [->, bend left=40] node {4} (b)
            edge [->, bend right=20] node [swap] {2} (f)
            edge [->, bend left=20] node {5} (a)
        (e) edge [->] node [swap] {4} (b)
            edge [->, bend left=20] node {2} (c)
        (f) edge [->, bend left=20] node {1} (c)
            edge [->] node [swap] {5} (e);
\end{tikzpicture}
\vspace{-3mm}
\caption{A network with the monetary concept.}
\label{Fig:ex}
\end{figure}
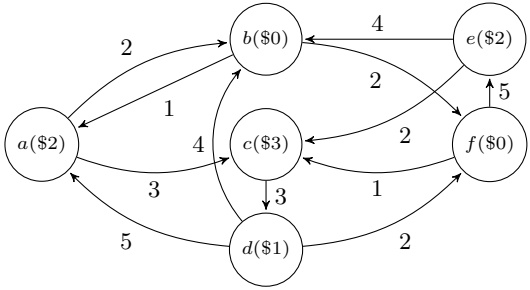

Specifically, we address the revenue maximization problem with a quantity constraint on commodities. 
Given a limited quantity of commodities and a social network with monetary concept regarding to the commodity, 
the problem is to determine the pricing of the commodity and 
identify a small group of people, 
i.e., a seed group, to be the initial customers receiving freebies to help promote the commodity at the beginning, 
so that the total revenue is maximized. 
In this paper, we investigate this problem on the monetary-concept incorporated LT-model \cite{Mirrokni2012}, and propose 
Algorithm PRUB (Pricing searching strategy using Revenue Upper Bound) to derive the optimal solutions. 
Then, we further revise PRUB to propose a heuristic Algorithm PRUB+IF (Pricing searching strategy using Revenue Upper Bound with Importance Feedback) for better efficiency.    
Experiments on real social networks show the good effectiveness of PRUB and PRUB+IF, demonstrating 
the revenue maximization with and without a quantity constraint differs from each other.

In summary, the contributions of this paper are:
\begin{enumerate}
\vspace{-1mm}
\item We are the first to address the general case of the revenue maximization problem. Note that the previous work of the revenue maximization without a quantity constraint is a special case of our problem with the quantity set as the total number of individuals on the social network.
\vspace{-1mm}
\item For the addressed problem, we proposed the optimal PRUB as well as a heuristic PRUB+IF.  
\vspace{-1mm}
\item Experiments on real social networks demonstrate the effectiveness of the proposed approaches for both general and special cases of the revenue maximization problem.
Note even in the case without a commodity constraint, the proposed approach outperforms the state of the art approach that incorporates the monetary concept \cite{Mirrokni2012}.
\end{enumerate}

\section{Problem formulation}
\label{Sec:PF}
To formulate the revenue maximization problem, in this section, we
first describe the social network with the monetary concept, then
explain how the social influences propagate over the network using the
monetary-concept incorporated LT-model \cite{Mirrokni2012}, and finally define the revenue and the revenue
maximization problem formally.

Given a commodity, a social network with the monetary concept,
referred to as a \textbf{monetizing social network}, is a weighted digraph
$G=(V,X,E,W,F)$, where $V$ is the set of
individuals and for each individual $v \in V$, his/her inherent valuation $\chi_v$ is carried in $X$. Usually, the valuation information $X$ can be estimated and learned from questionnaires or
historical sales data, with regard to the commodity \cite{Jung2005}. For any individuals $u,v\in V$, if 
$u$'s purchase will directly encourage $v$'s desire for the commodity, 
an edge from $u$ to $v$, denoted as
$e_{uv}\in E$ ($E$ standing for the edge set), represents the existance of the influence,   
and $w_{uv} \in W$ represents the influence strength as the weight on $e_{uv}$.
Adopting the Concave Graph Model in the work \cite{Mirrokni2012}, we consider a non-negative, non-decreasing, and concave function $F: \mathbb{R}^+ \rightarrow \mathbb{R}^+$ with $F(0)=0$ to transform the weights into the valuation concept. That is, given a set of individuals $S$ directly exerting influences on an individual $v$, $v$'s valuation towards a commodity will become $\chi_v + F(\sum_{i \in S} w_{iv})$.

Based on the monetary-concept incorporated LT-model \cite{Mirrokni2012}, the propagation of social influences over the
monetizing social network happens whenever individuals \textit{adopt} the
commodity. An individual \textit{adopts} the commodity if and only if
1) the pricing of the commodity is less than or equal to his/her
valuation, or 2) this individual is an initial customer receiving a
freebie. After an individual adopts the commodity, he/she will then
exert influences on his/her out-neighbors to encourage them and raise their valuations for adoption. 
The following shows an example of
social influences propagating over the monetizing social network.

\begin{example}
\label{Ex:ex1}
Given the monetizing social network in
Figure~\ref{Fig:ex} and the concave influence function $F(x)=x$, consider the promotion of a concert. Suppose that the pricing of the concert tickets is
$\$7$ and the seed group is $A = \{d\}$. Influenced by $d$, the valuations
of $a$, $b$, and $f$ are increased by $F(5)=\$5$, $F(4)=\$4$, and $F(2)=\$2$, respectively.
The valuations of $a$, $b$ and $f$ thus become $\$2+\$5=\$7$, $\$0+\$4=\$4$,
and $\$0+\$2=\$2$, respectively. Then, note that as the valuation of $a$
is equal to the pricing of the concert tickets $\$7$, $a$ also adopts the
concert ticket and exerts influences on its out-neighbors $b$ and $c$.
The valuations of $b$ and $c$ thus increase to $\$0+F(4+2)=\$6$ and
$\$3+F(3)=\$6$, respectively. After that, as no other individuals adopt
the concert ticket, the influence propagation stops. Consequently, under
$d$'s influence, $a$ will also adopt the concert ticket, and the
valuations of $b$, $c$ and $f$ become $\$6$, $\$6$, and $\$2$, respectively.

Consider another example of the seed group $A = \{d, f\}$. Under the
influence of $d$ and $f$, the valuations of $a$, $b$, $c$ and $e$ are
raised to $\$7$, $\$4$, $\$4$, and $\$7$, respectively. (Here $d$'s influence
on $f$ is neglected since $f$ as an initial customer has already adopted the concert ticket.)
Hence, $a$ and $e$ adopt the concert tickets and exert influences.
Consequently, the valuations of $b$ and $c$ become:
\vspace{-1mm}
\begin{displaymath}
\begin{array}{l}
b: \$0+F(4 \text{(from }d \text{)}+2 \text{(from }a \text{)}+4 \text{(from }e \text{)})=\$10 \\
c: \$3+F(1 \text{(from }f \text{)}+3 \text{(from }a \text{)}+2 \text{(from }e \text{)})=\$9
\end{array}
\vspace{-1mm}
\end{displaymath}
In the end, $a$, $b$, $c$, and $e$ all adopt the concert tickets since their
final valuations are larger than or equal to the pricing of the concert tickets
$\$7$.
\exend
\end{example}

Here we define that the revenue comes from the number of people
paying for the commodity and the pricing of the commodity. Recall Example~\ref{Ex:ex1} where the pricing is $\$7$ and the seed group is $\{d,f\}$, and
suppose that the quantity of the concert tickets is 4. The revenue is
$\$7\times 2=\$14$, since two concert tickets are used as free tickets given to
$d$ and $f$, and only two concert tickets are left for sale (even though
$a$, $b$, $c$, and $e$ all want to purchase these concert tickets).

Formally, given the quantity of commodities $n$, the pricing $p$ of the commodity, and the seed group $A$, the revenue is defined as follows.

\begin{definition}
\textbf{Revenue function.} Given the quantity $n$, the revenue at the pricing of $p$ with the seed group $A$ is
\begin{displaymath}
R(n,p,A) = p \times \min\{|\sigma(A) \setminus A|, n-|A|\},
\end{displaymath}
where $\sigma(A) \supseteq A$ is the set of individuals adopting the commodity under the influences of $A$.
\end{definition}

To increase the revenue, finding the proper pricing $p$ and identifying a proper seed group $A$ is important. Consider
the same setting mentioned above, but with an empty seed group and
the pricing set as $\$1$. The revenue is $R(4,\$1,\emptyset)=\$1\times 4=\$4$, which is
less than $R(4,\$7,\{d,f\})=\$14$. Since any company will expect to earn as higher revenue as
possible given a fixed amount of commodities, in this paper, we are
interested in looking for such pricing and a seed group as initial
customers that can bring maximum revenue. The proposed problem is
formally defined as follows.

\begin{definition}
\textbf{The RM$_\text{w/QC}$ Problem: Revenue Maximization with a Quantity Constraint. }Given a monetizing social network $G=(V,X,E,W,F)$, a set of input prices $P \subseteq \mathbb{R}^+$, and a quantity of commodities $n$, the problem is to determine the pricing
$p_{\text{max}} \in P$ of the commodity and find a seed group $A_{\text{max}} \subseteq V$ as initial customers, where $|A_\text{max}|\leq n$, such that the revenue 
$R(n,p,A)$ is maximized, i.e., $(p_{\text{max}},A_{\text{max}}) = \argmax\limits_{p,A}R(n,p,A)$.
\end{definition}

\vspace{-1mm}
Note that the revenue function $R(n,p,A)$ is not single-peaked with respect to the pricing, which means the approaches such as the binary search and the gradient decent search are not applicable to finding out the optimal pricing. 
For example, following the setting of Example~\ref{Ex:ex1} and letting the quantity of the concert tickets be 4, we obtain the highest revenue at pricing of $\$6$, $\$7$, and $\$8$ as $\$18$, $\$14$, and $\$16$, respectively (with the seed groups $\{d\}$, $\{d,f\}$, and $\{d,e\}$, accordingly).
Obviously, with respect to the pricing, $R(n,p,A)$ is not single-peaked.

Furthermore, the RM$_\text{w/QC}$ problem can be proved to be NP-hard. To show this, we first introduce a special case of RM$_\text{w/QC}$, and then prove that this special case is NP-hard. As the special case is NP-hard, RM$_\text{w/QC}$ is thus NP-hard.

\vspace{-1mm}
\begin{definition}
\textbf{SpecialRM. }This problem is a special case of
RM$_{\text{w/QC}}$, and asks for only one input price and a sufficient
quantity of commodities for the population. That is, given $G$,
$P=\{p\}$, and $n=|V|$, SpecialRM is to determine a seed group $A$
such that $R(|V|,p,A)$ is maximized.
\end{definition}

\vspace{-1mm}
\begin{theorem}
\label{Theo:np}
The problem of SpecialRM is NP-hard.
\end{theorem}

\begin{proof}
We reduce the Minimum Vertex Cover (MVC) problem to the SpecialRM
problem. Given every instance of MVC involving an undirected graph
$G' = (V',E')$, we can construct a corresponding instance
$G=(V,X,E,W,F)$ of SpecialRM as follows. (1) Set $V=V'$; (2) set $X
= \{\chi_v|\forall v\in V, \chi_v=0\}$; (3) direct all edges
in both directions, i.e., for each $e_{uv} \in E'$, $e_{uv} \in E$
and $e_{vu} \in E$; (4) for each $e_{uv}\in E$, define
$w_{uv}=\frac{p}{d_{\text{in}}(v)}$, where $d_{\text{in}}(v)$ is the
in-degree of $v$ in $G$; and (5) set $F(x)=x$. For the instance of
SpecialRM, if there exists a minimum set $A$ that maximizes the
revenue, we can obtain another set $A' = A \cup
(V\setminus\sigma(A))$. Then, it can be shown that $A'$ is the
minimum vertex cover of $G'$ since
\vspace{-1mm}
\begin{displaymath}
\sigma(A') \supseteq \sigma(A) \cup \sigma(V\setminus\sigma(A)) \supseteq \sigma(A) \cup (V\setminus\sigma(A)) = V.
\vspace{-1mm}
\end{displaymath}
Theorem~\ref{Theo:np} is thus proved.
\end{proof}

\vspace{-1mm}
\section{Algorithm}

\begin{table}[t]
\caption{\subref{T:maximalBudgets} The maximum valuations; \subref{T:maximalRevenues} the upper bounds of maximum revenue.}
\vspace{-1mm}
\centering
\subfigure[][]{
\small \centering 
\begin{tabular}{c|c}
\text{$v$} & \text{$X_\text{max}(v)$} \\ \hline 
$a$ & $\$8$\\ \hline 
$b$ & $\$10$\\ \hline 
$c$ & $\$9$\\ \hline 
$d$ & $\$4$\\ \hline 
$e$ & $\$7$\\ \hline
$f$ & $\$4$\\
\end{tabular}
\label{T:maximalBudgets}
}%
\subfigure[][]{
\small \centering
\begin{tabular}{c|c}
\text{$p$} & \text{$R_\text{bound}(4,p)$} \\ \hline 
$\$1$ & $\$4$\\ \hline
$\$2$ & $\$8$\\ \hline
$\$3$ & $\$12$\\ \hline
$\$4$ & $\$16$\\ \hline
$\$5$ & $\$20$\\ \hline
$\$6$ & $\$24$\\ \hline
$\$7$ & $\$28$\\ \hline
$\$8$ & $\$24$\\ \hline
$\$9$ & $\$18$\\ \hline
$\$10$ & $\$10$\\
\end{tabular}
\label{T:maximalRevenues}
}
\vspace{-2mm}
\end{table}

In this section, for the revenue maximization problem, we first
propose Algorithm PRUB (Pricing searching strategy using Revenue Upper
Bound) that is able to derive the optimal solutions. Note that
finding the optimal pair of $p_{\text{max}}$ and $A_{\text{max}}$ in
the revenue maximization problem is exhausted. Then, for better
efficiency, a heuristic Algorithm PRUB+IF (PRUB with Importance
Feedback) is introduced following the framework of PRUB.

\subsection{Algorithm PRUB}
\label{Sec:framework}

To tackle the RM$_\text{w/QC}$ problem, Algorithm PRUB is
designed to obtain the optimal solutions. As mentioned previously,
finding a proper pair of pricing and seed group is critical in this
problem, since the revenue comes from the pricing of the commodity and
the number of people paying for the commodity. A na\"{i}ve approach for
finding the optimal pair of pricing $p_{\text{max}}$ and seed group $A_{\text{max}}$ is to
search at all $p \in P$, enumerate all seed groups $A$, and calculate the
corresponding revenue $R(n,p,A)$. The rationale of Algorithm PRUB is based on
this na\"{i}ve approach and prunes the search space by 1) progressively
filtering out non-candidate pricing and 2) deriving an upper bound
of the size of seed groups for each price. We illustrate
the two ideas of pruning in the following.



\textbf{Non-candidate Pricing Filtering. }The first idea is to prune
the searching on non-candidate pricing. Motivated by the concern of only maximum revenue, this pruning derives an upper bound of maximum revenue at each price, and utilizes an achievable global revenue
$r_{\text{global}}$ that records the maximum revenue discovered so
far to progressively filter out the pricing that will not result in higher revenue. That is, if an upper
bound of maximum revenue at each price can be derived, PRUB then
does not need to seek the seed group for maximum revenue at a specific price when the upper bound at this price is less than or equal to $r_{\text{global}}$. When
$r_{\text{global}}$ is updated and increased progressively after the
successful searching of the seed group at each price, the
non-candidate pricing is also progressively detected and filtered
out. 

In order to infer the upper bound of maximum revenue at a specific price, the maximum number of individuals who have potential for adopting the commodity at this price is needed. To find out whether or not an individual has this potential, it needs to estimate the valuation of the individual.
Therefore, in the following, we first introduce the definitions of \textit{maximum valuations} and \textit{potential buyers} for defining the upper bound of maximum revenue.

\begin{definition}
\textbf{Maximum valuation. }The maximum valuation of an individual $v$ is the valuation under the influences from all $v$'s in-neighbors, which is denoted as
\begin{displaymath}
X_\text{max}(v) = \chi_v+F(\sum\limits_{u\in v\text{'s in-neighbors}}w_{uv}).
\end{displaymath}
\end{definition}

\begin{definition}
\textbf{Potential buyer. }An individual $v$ is regarded as a potential buyer at a specific price $p$ if $v$ has potential for adopting the commodity at $p$, i.e., $X_\text{max}(v) \geq p$.
\end{definition}

\begin{definition}
\textbf{Upper bound of maximum revenue. }Given a quantity constraint $n$, the upper bound of maximum revenue at a price $p$ is 
\begin{equation}
\label{E:upperRevenue}
R_{\text{bound}}(n,p) = p \times \min\{n, m_p\},
\end{equation}
where $m_p$ is the number of potential buyers at $p$, i.e., $m_p = |\{v
| v \in V, X_\text{max}(v) \geq p\}|$.

\end{definition}


\begin{example}
\label{Ex:framework}
Given the monetizing social network in Figure
\ref{Fig:ex}, the concave influence function $F(x)=x$, and a set of input prices $P = \{p|p \in \mathbb{I}^+, 1 \leq p \leq 10\}$, consider the promotion of a concert. Table~\ref{T:maximalBudgets} shows the maximum valuation
of each individual. Suppose the quantity of the concert tickets is set as 4. The
corresponding upper bounds of maximum revenue for all prices are listed in Table~\ref{T:maximalRevenues}.
\exend
\end{example}

In addition, to have better effectiveness of pruning by utilizing
the upper bound of maximum revenue, PRUB searches the prices in a descending
order of the upper bounds. Then, once PRUB discovers a price, at which the upper bound is less than or equal to the achievable global
revenue $r_{\text{global}}$, the searching for the following prices (including the current price) can be ignored.
Therefore, in the above example, PRUB will start the examination of
the prices from $p=\$7$, $p=\$6$, $p=\$8$, $\cdots$, and so on.


\textbf{Bound of Seed Group's Size. }The second idea for pruning
search space is to avoid enumerating useless seed groups under a
specific price. The inspiration is, under the consideration of a
specific price $p$, it involves more than $\frac{r_{\text{global}}}{p}$
quantities of commodities to be sold for finding the revenue
higher than the achievable global revenue $r_{\text{global}}$ obtained so far.
When considering the price $p$, Algorithm PRUB then bounds the size
of the seed groups as
\begin{equation}
\label{E:ABound}
\left|A\right|<n-\frac{r_{\text{global}}}{p},
\end{equation}
since the quantity of commodities is limited. Follow Example~\ref{Ex:framework}, and suppose the achievable global revenue $r_{\text{global}}=\$10$. In order to find the
revenue higher than $\$10$, when searching at the price $p=\$7$, PRUB
expects more than $\frac{10}{7}$ concert tickets left for sale and enumerates only the seed groups of sizes less than or equal to
2 (since the quantity of the concert tickets is set as 4). The searching on the seed groups of sizes 3 and 4 is thus pruned.

Note that incorporating the two pruning methods is satisfactory to the accuracy. Therefore, Algorithm PRUB is still able to derive the optimal solutions, even though the search space is reduced. The details of PRUB are presented in Algorithm~\ref{Algo:framework}.
In the following, we show an example to illustrate how PRUB figures out the optimal pair of $p_{\text{max}}$ and $A_{\text{max}}$.

\setlength{\algomargin}{2.5ex}
\begin{algorithm}[t]
\label{Algo:framework}
\caption{PRUB}
\scalefont{0.9}
\DontPrintSemicolon
\SetNlSty{texttt}{}{}
\SetAlgoNlRelativeSize{-3}
\SetNlSkip{0.25em}
\SetVlineSkip{1pt}
\KwIn{A monetizing social network $G=(V, X, E, W, F)$; a set of input prices $P$; a quantity of commodities $n$.}
\KwOut{The pricing $p_{\text{max}}$; the seed group $A_{\text{max}}$.}
$p_{\text{max}}\gets 0$, $A_{\text{max}}\gets\emptyset$, $r_{\text{global}} \gets 0$\;
Derive $R_{\text{bound}}(n,p)$ for all $p\in P$\;
Sort all $p \in P$ descendingly by $R_{\text{bound}}(n,p)$\;
\label{code:fram:search1}
\For{$p \in P$}{
	\If{$p$ is non-candidate pricing}{
		\KwRet{$p_{\text{max}}$, $A_{\text{max}}$}
	}
	Enumerate all the seed groups whose size is bounded by Equation~(\ref{E:ABound}) (including the size 0)\;
	Compute $R(n,p,A)$ for those enumerated seed groups\;
	\If{any $R(n,p,A)>r_{\text{global}}$}{
		\label{code:fram:update1}
			$p_{\text{max}}=p$, $A_{\text{max}}=A$, $r_{\text{global}}=R(n,p,A)$\;
		\label{code:fram:update2}
	}
}
\KwRet{$p_{\text{max}}$, $A_{\text{max}}$}
\end{algorithm}

\begin{example}
\label{Ex:exPRUB}
Follow Example~\ref{Ex:framework}, and initialize $p_{\text{max}}$, $A_{\text{max}}$, and $r_{\text{global}}$ as $\$0$, an
empty set, and $\$0$, respectively. According to the upper bounds of
maximum revenue in Table~\ref{T:maximalRevenues}, PRUB will consider the prices in the
following order: $p=\$7$, $p=\$6$, $p=\$8$, $\cdots$, and so on.

Beginning from the price $\$7$, PRUB checks whether or not
$R_\text{bound}(4,\$7) > r_{\text{global}}$. Since $\$28>\$0$, PRUB
looks for the maximum revenue at the price $\$7$ by enumerating all
the seed groups whose size is bounded by Equation~(\ref{E:ABound}) (including
the size 0). First, for the empty seed group, no individual adopts the
concert ticket at $\$7$. So PRUB goes to the size 1. Then, PRUB lists all the
seed groups of size 1 since $ 1 < 4 - \frac{0}{7} = 4$. Among all
the seed groups of size 1, PRUB finds the highest revenue $R(4,\$7,\{d\})=\$7$ (referred to Example~\ref{Ex:ex1}). As $\$7$ is higher than
$r_{\text{global}} = \$0$, PRUB updates $p_{\text{max}} = \$7$, $A_{\text{max}} = \{d\}$, and $r_{\text{global}} = \$7$. Before
seeking the seed groups of size 2 for maximum revenue, PRUB ensures
$ 2 < 4 - \frac{7}{7} = 3$. After all seed groups of size 2 are
enumerated, the highest revenue $R(4,\$7,\{d,f\})=\$14$ is found (referred to Example~\ref{Ex:ex1}). Because $\$14 >
r_{\text{global}} = \$7$, PRUB updates $p_{\text{max}} = \$7$, $A_{\text{max}} = \{d, f\}$, and $r_{\text{global}} = \$14$. Later, note
that $3 > 4 - \frac{14}{7} = 2$. The searching for the maximum revenue
at the price $\$7$ stops. The next price considered is $\$6$. A similar
process is performed until PRUB finds such a price $p$ that
$r_{\text{global}} \geq R_\text{bound}(4,p)$. In the end, the maximum
revenue is $R(4,\$6,\{d\})=\$18$, where the pair of $p_\text{max}=\$6$ and $A_\text{max}=\{d\}$ is the optimal solution.
\exend
\end{example}

\begin{theorem}
\label{Theo:complexity}
The time complexity of PRUB is $O(2^{|V|}|V|^2|P|)$.
\end{theorem}
\begin{proof}
We first show that PRUB costs $O(2^{|V|}|V|^2)$ time at a specific price for
searching the maximum revenue. Given the quantity of commodities
$n$, the largest size of the seed group is $n$. In total, there are
thus $\sum\limits_{i=0}^{n}{{|V|}\choose i}$ seed groups. For a seed group
of size $i$, there are $i$ seeds exerting influences on the other O($|V|
- i$) individuals that are not seeds. Due to the influence cascade, it
costs O($(|V| - i)^2$) time to figure out the influence spread of
a seed group of size $i$. Therefore, the total time complexity for
searching the maximum revenue by enumerating seed groups is
\begin{displaymath}
\begin{array}{rl}
&O(\sum\limits_{i=0}^{n}{{{|V|}\choose{i}}(|V|-i)^2}) \\
=&O(\sum\limits_{i=0}^{n}{\frac{|V|!}{(|V|-i)!i!} (|V|-i)^2}) \\
=&O(\sum\limits_{i=0}^{n}{\frac{|V|^2(|V|-1)!}{(|V|-i-1)!i!}-\frac{|V|(|V|-1)(|V|-2)!}{(|V|-i-1)!(i-1)!}}) \\
=&O(|V|^2\sum\limits_{i=0}^{n}{{|V|-1}\choose{i}}-|V|^2\sum\limits_{i=0}^{n}{{|V|-2}\choose{i-1}}+|V|\sum\limits_{i=0}^{n}{{|V|-2}\choose{i-1}}) \\
=&O(2^{|V|-1}|V|^2-2^{|V|-2}|V|^2+2^{|V|-2}|V|) \\
=&O(2^{|V|}|V|^2).\\
\end{array}
\end{displaymath}
As searching at a price costs $O(2^{|V|}|V|^2)$, the total time complexity of
PRUB is thus $O(2^{|V|}|V|^2|P|)$. Theorem~\ref{Theo:complexity} is proved.
\end{proof}

Note that the high complexity of Algorithm PRUB indicates the
poor scalability of the optimal algorithm for real social networks
which are usually large. Therefore, in the following subsection, we
propose a heuristic based on the framework of PRUB to obtain the
feasible solutions on larger instances.

\subsection{Algorithm PRUB+IF}
\label{Sec:contribution}

In this section, the heuristic algorithm PRUB+IF (PRUB with Importance
Feedback) is proposed as a feasible solution. The
heuristic PRUB+IF differs from PRUB at the way of finding the most
proper seed group at each price. Following the framework of PRUB,
PRUB+IF introduces the concept of \textit{pricing-sensitive importance},
accompanying the contribution feedback from the out-neighbors, for
selecting seeds in a heuristic manner, instead of listing all seed
groups.

The main idea of PRUB+IF is that, an individual with greater
potential for making others adopt the commodity should be regarded as
more important. Selecting the most important individuals as seeds
greedily is able to lead to a feasible solution. Then, the question
is how to evaluate one's potential for making others adopt the commodity. The
intuition is about how many others will be encouraged and how much
their valuations can be increased for approaching the pricing, under
one's influence. We then accumulate these effects to calculate one's
importance. Here only the effects on potential buyers should be
included in the accumulation since only the potential buyers have possibilities of paying for the commodity. In addition, note that the
individuals who newly adopt the commodity under one's (direct and
indirect) influence will further spread their influences to encourage
more others in adoption. In order to estimate one's importance
more carefully, the effects through influence cascades should also be
included in the accumulation. Therefore, PRUB+IF introduces the
pricing-sensitive importance that sums up the consideration of the
normalized weights, the influence propagation, and the potential
buyers in the measurement of one's advantage in the commodity
promotion at the given pricing. Then, the strategy is to select seeds in
accordance with the pricing-sensitive importance. 

In the following,
we introduce the three key points, Normalized Weight, Feedback of Influence Propagation, and Potential-Buyer Filtering, in measuring the pricing-sensitive importance. Briefly, for each individual $u$, 
1) the Normalized Weight is used to evaluate $u$'s importance according to $u$'s direct effect on another individual $v$'s adoption; 2) the Feedback of Influence Propagation is incorporated to evaluate $u$'s importance towards $v$ by considering also the indirect effects through influence cascades; 3) the Potential-Buyer Filtering is considered to derive $u$'s pricing-sensitive importance from accumulating all $u$'s direct and indirect effects on all the other potential buyers.


\begin{enumerate}

\vspace{-1mm}
\item \textbf{Normalized Weight. }

By intuition, for commodity promotion, 
an individual $u$'s importance towards another individual $v$ 
can be evaluated from how much $v$'s valuation approaches the pricing due to $u$'s effect. 
Inspired by this intuition, PRUB+IF 
calculates the normalized weight $\hat{w}_{uv}$, indicating $u$'s importance towards $v$'s adoption, as follows. 
\vspace{-2mm}
\begin{flalign*}
\footnotesize
&\hat{w}_{uv} =&
\end{flalign*}
\begin{equation}
\label{E:nW}
\footnotesize
\begin{cases}
0, &\text{if } p \leq X_{A}(v) \\
\min\{1, \frac{F(w_{uv}+\sum\limits_{i\in \sigma(A)}w_{iv})-F(\sum\limits_{i\in \sigma(A)}w_{iv})}{ p-X_{A}(v)}\}, &\text{otherwise}
\end{cases}
,
\end{equation}

where $X_{A}(v) = \chi_v+F(\sum\limits_{i\in \sigma(A)}w_{iv})$ is $v$'s valuation of the commodity under the current seed group $A$'s effects. In the case of $p \leq X_{A}(v)$, it means $v$ has adopted the commodity. So the normalized
weight $\hat{w}_{uv} = 0$, indicating that $u$'s importance towards $v$ is none due to no impact from $u$ on $v$'s adoption of the commodity.
For the other case, the normalized
weight from $u$ to $v$ is bounded by 1, and is the portion of $u$'s impact on $v$'s
adoption of the commodity.

\vspace{-2mm}
\begin{example}
\label{Ex:Wnorm}
Given the monetizing social network in Figure~\ref{Fig:ex} and the concave influence function $F(x)=x$, consider the promotion of a concert. Given
$p= \$7$ and $A=\{\}$, no one has adopted the concert tickets. Under this
condition, $X_{\{\}}(a)=\$2$ is less than the pricing of $\$7$, and  
the normalized weight from $d$ to $a$ is 
$\hat{w}_{da} = \min\{1, \frac{\$5-\$0}{\$7-\$2} \} = 1$.
This means that selecting $d$ into the seed group $\{\}$ can make $a$ adopt the concert ticket immediately. 
Similarly,
\vspace{-1mm}
\begin{displaymath}
\begin{array}{llll}
\hat{w}_{db} = \frac{4}{7}, &
\hat{w}_{dc} = 0, &
\hat{w}_{de} = 0, &
\hat{w}_{df} = \frac{2}{7}. \exend\\
\end{array}
\end{displaymath}
\end{example}

\item \textbf{Feedback of Influence Propagation. }

Note that the effect of influence propagation is important in the commodity promotion. 
To carefully estimate one's advantage in the promotion, PRUB+IF further takes the importance feedback from those individuals $i$ into account on $u$'s importance, if $i$ adopts the commodity due to $u$'s direct or indirect effect and thus devotes to the propagation of $u$'s effect. 
Therefore, $u$'s importance towards $v$ will be derived recursively as follows. 
\begin{flalign*}
\footnotesize
&IF^{(k+1)}(u,v)=&
\end{flalign*}
\vspace{-3mm}
\begin{equation}
\label{E:IF}
\footnotesize
\begin{cases}
\min\{1,  IF^{(k)}(u,v)  +\sum\limits_{i\in V_u^{(k)}}{\hat{w}_{iv}}\}, & \text{if $u\neq v$}\\
0, & \text{otherwise}
\end{cases}
,
\end{equation}
where $IF^{(0)}(u,v)=\hat{w}_{uv}$ and $V_u^{(k)}$ contains the individuals who newly adopt the commodity by $u$'s effects at $k^\text{th}$ propagation, i.e., $V_u^{(k)}=\{i|i\in V, IF^{(k)}(u,i)=1\wedge\forall k'<k, IF^{(k')}(u,i)<1\}$.
As the influence propagates, there are multiple updates for $u$'s importance towards $v$.
When the propagation stops, i.e., $V_u^{(k)}=\emptyset$, the final importance of $u$ towards $v$ is denoted as $IF(u,v)$.
In other words, Equation (\ref{E:IF}) evaluates the advantage of
selecting $u$ into the current seed group from the aspect of $v$'s adoption, by considering the direct effect $IF^{(0)}(u,v)=\hat{w}_{uv}$ derived by Equation (\ref{E:nW}) and the indirect effects $\sum\limits_{k}\sum\limits_{i\in V_u^{(k)}}{\hat{w}_{iv}}$ propagated through individuals $i$. 
In addition, Equation (\ref{E:IF})
ensures that the importance of $u$ towards $v$ is bounded by 1, even though there may be many direct and indirect impacts from $u$ to $v$.

\vspace{-2mm}
\begin{example}
\label{Ex:propagate}
Follow Example~\ref{Ex:Wnorm} and consider the individual $d$. Initially, $V_d^{(0)} = \{a\}$, because $IF^{(0)}(d,a)=\hat{w}_{da}=1$.
$IF^{(1)}(d,\cdot)$ for all individuals (excluding
$d$ itself) are thus calculated as follows.
\vspace{-1mm}
\begin{displaymath}
\begin{array}{l}
IF^{(1)}(d, a) = \min \{1,  1 + \hat{w}_{aa} \} = 1  \\ 
IF^{(1)}(d, b) = \min \{1, \frac{4}{7} + \hat{w}_{ab} \} =  
\frac{6}{7} \\
IF^{(1)}(d, c) = \min \{1, 0 + \hat{w}_{ac}\} =
\frac{3}{4} \\
IF^{(1)}(d, e) = \min \{1, 0 + \hat{w}_{ae}\} = 0\\
IF^{(1)}(d, f) = \min \{1, \frac{2}{7} + \hat{w}_{af}\} = \frac{2}{7} \\
\end{array}
\vspace{-1mm}
\end{displaymath}
Since no other individuals newly adopt the concert ticket, i.e., $V_d^{(1)}=\emptyset$, the propagation stops. Finally, $IF(d,\cdot)=IF^{(1)}(d,\cdot)$ are as above. \exend
\end{example}

\item \textbf{Potential-Buyer Filtering. }

Once the advantage of selecting an individual as a seed can be
estimated from the aspect of another's adoption, intuitively, the
importance of an individual can be derived by summing up his/her
importance towards all the others. However, counting one's importance
towards all the others may be misleading, since the advantage estimated from those who will never adopt the commodity at this price is also counted. To avoid overestimating the
importance of an individual in the promotion, PRUB+IF thus incorporates the
concept of potential buyers into the pricing-sensitive importance which is 
formally defined below.
\vspace{-1mm}
\begin{equation}
\label{E:psi}
\Psi(u) = \sum\limits_{i\in V_{\text{potential}}}{IF(u,i)}
\vspace{-1mm}
\end{equation}
where $V_{\text{potential}}$ is the set of potential buyers at this price, i.e., the individuals
whose maximum valuation is larger than or equal
to the price (the same as introduced in Section~\ref{Sec:framework}). As a result, only if an individual is a potential buyer who has potential for adopting the commodity at the price, the increase of his/her valuation
will be regarded as important.

\begin{example}
\label{Ex:psi}
Follow Example~\ref{Ex:propagate}. Since the potential buyers at the pricing of $\$7$ are $a$, $b$, $c$, and $e$ (referred to Table~\ref{T:maximalBudgets}), the pricing-sensitive
importance of $d$ is $\Psi(d) = 1+\frac{6}{7}+\frac{3}{4}+0=\frac{73}{28}$.
\exend
\end{example}


\end{enumerate}

\setlength{\algomargin}{2.5ex}
\begin{algorithm}[t]
\caption{PRUB+IF}
\label{Algo:contribution}
\scalefont{0.9}
\DontPrintSemicolon
\SetNlSty{texttt}{}{}
\SetAlgoNlRelativeSize{-3}
\SetNlSkip{0.25em}
\SetVlineSkip{1pt}
\KwIn{A monetizing social network $G=(V, X, E, W, F)$; a set of input prices $P$; a quantity of commodities $n$.}
\KwOut{The pricing $p_{\text{max}}$; the seed group $A_{\text{max}}$.}
$p_{\text{max}}\gets 0$, $A_{\text{max}}\gets\emptyset$, $r_{\text{global}} \gets 0$\;
\label{code:fram:init1}
Derive $R_{\text{bound}}(n,p)$ for all $p\in P$\;
Sort all $p \in P$ descendingly by $R_{\text{bound}}(n,p)$\;
\label{code:fram:init2}
\For{$p \in P$}{
	\If{$p$ is non-candidate pricing}{
	\label{code:fram:prune4}
		\KwRet{$p_{\text{max}}$, $A_{\text{max}}$}
	\label{code:fram:prune5}
	}
	$A \gets \emptyset$\;
	\label{code:fram:empty1}
	Compute $R(n,p,A)$\;
	\label{code:fram:empty3}	
	\If{$R(n,p,A)>r_{\text{global}}$}{
	\label{code:fram:update5}
		$p_{\text{max}}=p$, $A_{\text{max}}=A$, $r_{\text{global}}=R(n,p,A)$\;
	\label{code:fram:empty2}
	}
	\While{$|A| < n-\frac{r_{\text{global}}}{p}$}{
	\label{code:fram:prune6}
		\For{$u\in V$ has not adopted the commodity}{
		\label{code:fram:differ3}
			Compute $\Psi(u)$ according to Equation (\ref{E:psi})\;
			\label{code:fram:psi}
		}
		$s\gets$ the individual with the greatest $\Psi(\cdot)$	\;
		$A \gets A\cup\{s\}$ \;
		\label{code:fram:differ4}
		Compute $R(n,p,A)$\;
		\If{$r>r_{\text{global}}$}{
		\label{code:fram:update3}
		$p_{\text{max}}=p$, $A_{\text{max}}=A$, $r_{\text{global}}=R(n,p,A)$\;		\label{code:fram:update4}
		}
	}	
}
\KwRet{$p_{\text{max}}$, $A_{\text{max}}$}
\end{algorithm}

The algorithm of PRUB+IF is presented in
Algorithm~\ref{Algo:contribution}. PRUB+IF uses the framework of
PRUB as a base, and differs from PRUB at the strategy for seed
selection (Algorithm~\ref{Algo:contribution} Lines~\ref{code:fram:differ3}-\ref{code:fram:differ4}). 
In the following, we go through the details of PRUB+IF.
Similarly to PRUB, PRUB+IF initializes $p_\text{max}$, $A_\text{max}$, and $r_{\text{global}}$, compute the upper
bounds of maximum revenue for all prices, and then considers the prices
one-by-one in a descending order of the upper bounds of maximum revenue
(Lines~\ref{code:fram:init1}-\ref{code:fram:init2}). At a specific price $p$, PRUB enumerates all seed
groups while PRUB+IF iteratively picks up a seed until the size of the seed group $A$ reaches the bound in Equation (\ref{E:ABound}) (Line~\ref{code:fram:prune6}).
Specifically, when considering a price $p$, PRUB+IF first checks
whether or not $p$ is non-candidate pricing that has no chance to yield higher revenue than the achievable global revenue
$r_{\text{global}}$ (Line~\ref{code:fram:prune4}).
If $p$ passes
the pruning, PRUB+IF first calculates the revenue regarding to the
empty seed group (Lines~\ref{code:fram:empty1}-\ref{code:fram:empty3}). After that, PRUB+IF expands the
seed group by one seed at a time according to the pricing-sensitive
importance $\Psi(\cdot)$ in Equation
(\ref{E:psi}) (Line~\ref{code:fram:psi}). 
Each time the individual who has not yet adopted the commodity and has the greatest pricing-sensitive
importance will be selected into the current seed group (Lines~\ref{code:fram:differ3}-\ref{code:fram:differ4}).
In the case of the newly obtained revenue higher than $r_{\text{global}}$, PRUB+IF updates $p_{\text{max}}$, $A_{\text{max}}$, and $r_{\text{global}}$ similarly to PRUB (Lines~\ref{code:fram:update3}-\ref{code:fram:update4}).
Once the size of the seed
group reaches the bound in Equation (\ref{E:ABound}) (Line~\ref{code:fram:prune6}), PRUB+IF
considers the next price. After all prices are examined,
$p_{\text{max}}$ and $A_{\text{max}}$ are returned as the solution.
The following example demonstrates the searching process of PRUB+IF.

\renewcommand{\arraystretch}{1.2}

\begin{example}

Follow Example~\ref{Ex:framework}, and initialize $p_{\text{max}}$, $A_{\text{max}}$, and $r_{\text{global}}$ as $\$0$, an
empty set, and $\$0$, respectively.
According to the upper bounds of maximum revenue in
Table~\ref{T:maximalRevenues}, the first price considered is $\$7$.
Since $R_\text{bound}(4,\$7) = \$28 > r_{\text{global}} = \$0$,
PRUB+IF tries to seek for higher revenue at the price $\$7$. First, for $A =
\{ \}$, there is no revenue since the inherent valuations of all
individuals are less than $\$7$. After that, for selecting the first seed, PRUB+IF 
computes the pricing-sensitive importance
of all individuals in a similar way as shown in Example~\ref{Ex:psi}:
\vspace{-1mm}
\begin{displaymath}
\begin{array}{lll}
\Psi(a) = \frac{29}{28}, &
\Psi(b) = \frac{1}{5}, &
\Psi(c) = 0, \\
\Psi(d) = \frac{73}{28}, &
\Psi(e) = \frac{15}{14}, &
\Psi(f) = \frac{65}{28}. \\
\end{array}
\vspace{-1mm}
\end{displaymath}
Hence, the first seed selected is $d$ and the corresponding revenue
will be $R(4,\$7,\{d\})=\$7$ (referred to Example~\ref{Ex:ex1}).
As $\$7 > r_{\text{global}} = \$0$, PRUB+IF performs the updates of
$p_{\text{max}}=\$7$, $A_{\text{max}}=\{d\}$, and $r_{\text{global}}=\$7$. Until now, the process for the first seed
finishes, and PRUB+IF checks whether or not to continue selecting
the second seed.
Since $|A|=|\{d\}| =1 < n-\frac{r_{\text{global}}}{p}=4-\frac{7}{7}=3$, the
pricing-sensitive importance of all individuals, except the
individuals who have adopted the concert tickets, i.e., $a$ and $d$, is
calculated with respect to the current seed group $A=\{d\}$: 
\vspace{-1mm}
\begin{displaymath}
\begin{array}{llll}
\Psi(b) = 0, &
\Psi(c) = 0, &
\Psi(e) = 2, &
\Psi(f) = 3. \\
\end{array}
\vspace{-1mm}
\end{displaymath}
Accordingly, $f$ will be selected as the second seed and the revenue
is $R(4,\$7,\{d,f\})=\$14$ (referred to Example~\ref{Ex:ex1}),
which leads to the updates of $p_{\text{max}}=\$7$, $A_{\text{max}}=\{d,f\}$, and $r_{\text{global}}=\$14$. Until now, the
process for the second seed finishes, and PRUB+IF checks whether or not to continue selecting the third seed.
However, $|A|=2$ reaches the bound
$n-\frac{r_{\text{global}}}{p}=4-\frac{14}{7}=2$ so that the searching
at the price $\$7$ stops.

The searching at the remaining prices $p=\$6$, $p=\$8$, $\cdots$, and so on, is performed in
the same manner as that at the price $\$7$.
In the end, PRUB+IF will report the pricing of $\$6$ and the seed group
$\{d\}$ as the answer. \exend
\end{example}

\begin{theorem}
\label{Theo:complexity2}
The time complexity of PRUB+IF is $O(n|V|^3|P|)$.
\end{theorem}
\begin{proof}
We first show that PRUB+IF costs $O(n|V|^3)$ time at a specific price for
searching the maximum revenue. Given the quantity of commodities
$n$, the largest size of the seed group is $n$. 
When $i$ seeds have been selected, it takes $O((|V|-i)^2)$ time to compute the pricing-sensitive importance for each of the other $O(|V|-i)$ nodes.
Therefore, the total time complexity for
searching the maximum revenue is
\begin{displaymath}
O(\sum\limits_{i=1}^{n-1}{(|V|-i)^3})=O(n|V|^3).
\end{displaymath}
As searching at a price costs $O(n|V|^3)$, the total time complexity of
PRUB+IF is thus $O(n|V|^3|P|)$. Theorem~\ref{Theo:complexity2} is proved.
\end{proof}

\vspace{-5mm}
\section{Experiments}

\begin{table*}[t]
\vspace{-5mm}
\caption{Summaries of real datasets and the parameters of each valuation distribution.}
\centering
\small
\label{T:summary}
\begin{tabular}{c||rrrrr||cc||cc|cc}
\multirow{2}{*}{Dataset} & \multirow{2}{*}{\#nodes} & \multirow{2}{*}{\#edges} & \multirow{2}{*}{Avg. degree} & \multirow{2}{*}{Avg. weight} & \multirow{2}{*}{Avg. clustering coeff.} & \multicolumn{2}{c||}{Normal} & \multicolumn{4}{c}{M-shape}\\ 
&&&&&& $\mu$ & $\sigma^2$ & $\mu$ & $\sigma^2$ & $\mu$ & $\sigma^2$ \\ \hline
\textit{highschool} & 50 & 140 & 5.60 & 37.29 & 0.4754 & 10 & 8.16 & 4 & 1.78 & 16 & 1.78\\
\textit{digg} & 30,360 & 85,155 & 5.61 & 6.07 & 0.0053 & 5 & 2.04 & 2 & 0.44 & 8 & 0.44 \\
\textit{facebook} & 45,813 & 183,412 & 8.01 & 4.66 & 0.1106 & 5 & 2.04 & 2 & 0.44 & 8 & 0.44 \\
\end{tabular}
\vspace{-6mm}
\end{table*}

In this section, we first provide the evidence that the revenue maximization problem with a quantity constraint differs from that without a quantity constraint. Besides, we also show the approximation performance of PRUB+IF to PRUB. In order to demonstrate the effectiveness of PRUB+IF, we further conduct experiments on larger real social networks and compare PRUB+IF with other heuristics. All the programs are implemented in Java with version
1.6.0\_27. The experiments are performed on a quad-core Intel$^\circledR$ Xeon$^\circledR$
X5450 3.00GHz PC with 8GB RAM using Ubuntu 8.04.2.

{\bf{Datasets.}} We use three real social networks as follows. 
1) The dataset \textit{highschool} (sampled from the dataset \cite{Coleman1964} due to poor scalability of PRUB and the compared approach) is a friendship network for the boys in a small highschool in Illinois.\footnote{\url{http://moreno.ss.uci.edu/data.html\#high}} The nodes are boys in the highschool, and each directed edge represents a boy choosing another as a friend in questionnaires. 
2) The dataset \textit{digg} \cite{Choudhury2009} is a reply network of Digg.\footnote{\url{http://www.public.asu.edu/~mdechoud/datasets.html} (The reply system has been out of service since 2012, but the data were collected in 2008.)} Each node is a user in Digg and each directed edge stands for a user replying to another.
3) The dataset \textit{facebook} \cite{Viswanath2009} is the communication network of Facebook.\footnote{\url{http://socialnetworks.mpi-sws.org/data-wosn2009.html}} The nodes are the users in
Facebook while each directed edge represents that a user posts an
article on another user's wall. 
For these datasets, the self-loops and isolated nodes are removed, and the edge weights are derived from the number of choices, replies, and articles of one to another, respectively. 
Both the degree distributions and weight distributions of all the datasets follow the well-known power law \cite{Clauset2009}. 
More details are shown in the left five columns of Table~\ref{T:summary}. 

{\bf{Valuation Distributions.}} Note that there is lack of valuation
information in these publicly available datasets. A na\"{i}ve way is to learn the valuation information through questionnaires or historical sales data toward a specific commodity. 
However, note that the valuation distributions vary with commodities. 
In this paper, we then explore the performance
under different valuation distributions as follows.

\begin{itemize}
\vspace{-2mm}
\item Normal distribution (N): Most people have similar valuations and a few people have extremely higher or lower valuations toward a commodity.
\vspace{-1mm}
\item M-shape distribution (M) \cite{Ohmaa2006}: People do not have consensus on the value of a commodity. One group of people have higher valuations and the other group of people have lower valuations toward the commodity. 
\end{itemize}
\vspace{-1mm}
In our experiments, we implement the M-shape distribution using two Normal distributions, where the valuation of each individual follows one of the Normal distributions with equal probabilities. 
The parameters for generating the valuations are listed in the right three columns of Table~\ref{T:summary}.

\textbf{Concave influence function.} As introduced in Section~\ref{Sec:PF}, any non-negative, non-decreasing, and concave function can be applied to the concave influence function. In our experiments, we consider the function $F(x)=x$ to transform the weights into the valuation concept for simplicity.

\textbf{Input prices.} 
By considering different currencies, costs, and cultures in different countries, the pricing of a commodity may be different. 
For simplicity, we set the prices as input parameters in the experiments. 
For \textit{highschool}, the input prices are all integers (as the unit of the dollars) in $[1, 300]$, due to the smaller size of the network. For \textit{digg} and \textit{facebook}, the input prices are all integers in $[1, 2000]$.

\subsection{Consideration of Quantity Constraints}
\label{Sec:exp1}

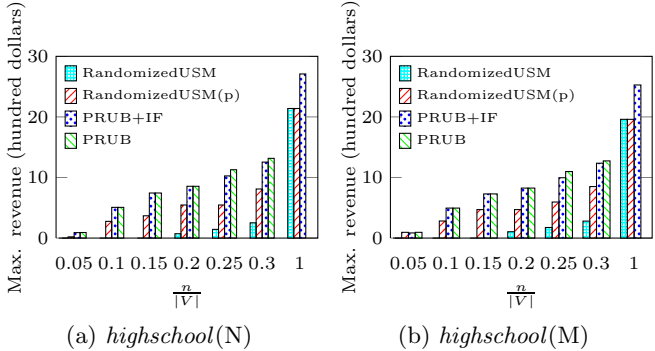
\begin{figure}[t]
\centering
\subfigure[][\textit{highschool}(N)]{
\centering
\label{Fig:smallN}
\begin{tikzpicture}[-]
\scriptsize
\begin{axis}[
ybar=0pt,
enlarge x limits=0.08,
bar width=0.08cm,
legend cell align=left,
legend image code/.code={\draw[#1] (0cm,-0.1cm) rectangle (0.15cm,0.1cm);}, 
legend style={at={(0,1), font=\tiny},anchor=north west,legend columns=1, draw=none, fill=none},
ylabel style={text width=4cm, align=center, yshift=-0.7cm},
ylabel=Max. revenue (hundred dollars),
symbolic x coords={0.05, 0.1, 0.15, 0.2, 0.25, 0.3, 1},
xlabel style={align=center},
xlabel={$\frac{n}{|V|}$},
xtick=data,
x tick style={opacity=0},
ymax=30,
ymin=0,
x=0.5cm,
height=4cm
]
\addplot
[pattern=my grid, pattern color=cyan]
 coordinates {
(0.05, 0)
(0.1, 0)
(0.15, 0)
(0.2, 0.72)
(0.25, 1.44)
(0.3, 2.52)
(1, 21.39)
};
\addplot
[pattern=north east lines, pattern color=red]
 coordinates {
(0.05, 0.20)
(0.1, 2.76)
(0.15, 3.68)
(0.2, 5.46)
(0.25, 5.46)
(0.3, 8.10)
(1, 21.39)
};
\addplot
[pattern=crosshatch dots,pattern color=blue]
 coordinates {
(0.05, 0.90)
(0.1, 5.07)
(0.15, 7.44)
(0.2, 8.55)
(0.25, 10.26)
(0.3, 12.54)
(1, 27.09)
};
\addplot
[pattern=north west lines, pattern color=green]
 coordinates {
(0.05, 0.92)
(0.1, 5.07)
(0.15, 7.44)
(0.2, 8.55)
(0.25, 11.28)
(0.3, 13.16)
};
\legend{RandomizedUSM, RandomizedUSM(p), PRUB+IF, PRUB}
\end{axis}
\end{tikzpicture}
}%
\subfigure[][\textit{highschool}(M)]{
\centering
\label{Fig:smallM}
\begin{tikzpicture}[-]
\scriptsize
\begin{axis}[
ybar=0pt,
enlarge x limits=0.08,
bar width=0.09cm,
legend cell align=left,
legend image code/.code={\draw[#1] (0cm,-0.1cm) rectangle (0.15cm,0.1cm);}, 
legend style={at={(0,1), font=\tiny},anchor=north west,legend columns=1, draw=none, fill=none},
ylabel style={text width=4cm, align=center, yshift=-0.7cm},
ylabel=Max. revenue (hundred dollars),
symbolic x coords={0.05, 0.1, 0.15, 0.2, 0.25, 0.3, 1},
xlabel style={align=center},
xlabel={$\frac{n}{|V|}$},
xtick=data,
x tick style={opacity=0},
ymax=30,
ymin=0,
x=0.5cm,
height=4cm
]
\addplot
[pattern=my grid, pattern color=cyan]
 coordinates {
(0.05, 0)
(0.1, 0)
(0.15, 0)
(0.2, 1.05)
(0.25, 1.75)
(0.3, 2.80)
(1,19.6)
};
\addplot
[pattern=north east lines, pattern color=red]
 coordinates {
(0.05, 0.95)
(0.1, 2.82)
(0.15, 4.70)
(0.2, 4.70)
(0.25, 5.95)
(0.3, 8.50)
(1,19.6)
};
\addplot
[pattern=crosshatch dots,pattern color=blue]
 coordinates {
(0.05, 0.85)
(0.1, 4.95)
(0.15, 7.29)
(0.2, 8.25)
(0.25, 9.96)
(0.3, 12.36)
(1,25.27)
};
\addplot
[pattern=north west lines, pattern color=green]
 coordinates {
(0.05, 0.95)
(0.1, 4.95)
(0.15, 7.29)
(0.2, 8.25)
(0.25, 10.98)
(0.3, 12.74)
};
\legend{RandomizedUSM, RandomizedUSM(p), PRUB+IF, PRUB}
\end{axis}
\end{tikzpicture}
}%
\vspace{-4mm}
\caption{The comparison for maximum revenue.}
\label{Fig:small}
\end{figure}

In this subsection, we first show that the consideration of a quantity constraint in the revenue maximization problem makes it differ from the original problem by comparing PRUB with the previous work \cite{Mirrokni2012}, referred to as RandomizedUSM in this paper. Due to poor scalability of PRUB and RandomizedUSM, these experiments are conducted on the small dataset \textit{highschool} with Normal and M-shape valuation distributions. 
In addition, note that since RandomizedUSM does not consider a quantity constraint, we derive the revenue with two kinds of post-processes as follows. 

\begin{itemize}
\item RandomizedUSM: Given the solution pair of $p$ and $A$ RandomizedUSM reports, 
this post-process considers the quantity constraint $n$ by calculating the quantities left for sale and computing the revenue as $p \times \min\{|\sigma(A)\setminus A|, n-|A|\}$, where $\sigma(A) \supseteq A$ is the set of individuals adopting the commodity under $A$'s effects.
\item RandomizedUSM(p): This post-process considers the quantity constraint $n$ at each price and then returns the highest revenue among all prices. For each price $p$ and the corresponding seed group $A_p$ obtained by RandomizedUSM, the revenue is $p \times \min\{|\sigma(A_p)\setminus A_p|, n-|A_p|\}$, where $\sigma(A_p) \supseteq A_p$ is the set of individuals adopting the commodity under $A_p$'s effects.
\end{itemize}

Figure~\ref{Fig:small} shows the maximum revenue reported by RandomizedUSM, RandomizedUSM(p), PRUB+IF and PRUB with respect to different ratios of the quantity of commodities to the number of population, i.e., $\frac{n}{|V|}$. 
Accordingly, the revenue reported by 
RandomizedUSM and RandomizedUSM(p) is, on average, less than 70\% of that obtained by PRUB. 
Moreover, note that RandomizedUSM has the poorest performance and even brings zero revenue at the cases of the supply ratio below 0.2. 
These results demonstrate that the capability of the previous approach \cite{Mirrokni2012} (even with the post-processes) is limited for the RM$_\text{w/QC}$ problem.
On the other hand, even for the special case that the quantity is not constrained, i.e., the ratio of the quantity of commodities to the number of population $\frac{n}{|V|}=1$, PRUB+IF still outperforms RandomizedUSM and RandomizedUSM(p). 
This is because PRUB+IF selects seeds in order of the pricing-sensitive importance while RandomizedUSM and RandomizedUSM(p) consider the seeds in arbitrary order.

On the same datasets, we then also explore the approximation of PRUB+IF to PRUB (which derives the optimal solutions). According to the results in Figure~\ref{Fig:small}, the performance of PRUB+IF is very close to that of PRUB, i.e., about 96\% on average, and outperforms both RandomizedUSM and RandomizedUSM(p). 
This implies that the proposed pricing-sensitive improtance of PRUB+IF is highly effective to identity the crucial individuals as seeds.  

\subsection{Performance of the PRUB+IF heuristic}

In this subsection, we conduct experiments on larger datasets \textit{digg} and \textit{facebook} to further show the performance of Algorithm PRUB+IF in terms of effectiveness, pricing, and efficiency. 
Here PRUB+IF is compared with two other heuristics PRUB+R and PRUB+SW. These two heuristics also take the PRUB approach as the framework while applying different strategies for seed selection. Specifically, PRUB+R (R standing for Random) selects seeds randomly; PRUB+SW (SW meaning Sum of the Weights) always picks up the individual with the maximum sum of out-weights from those who have not yet adopted the commodity.
Besides, in order to demonstrate the advantage of each part of PRUB+IF, the heuristics PRUB+N, PRUB+F, and PRUB+P represent the approaches that incorporate Normalized Weight, Feedback of Influence Propagation, and Potential-buyer Filtering, respectively, for seed selection.

\begin{figure*}[t]
\vspace{-5mm}
\centering
\subfigure[][\textit{digg}(N)]{
\centering
\label{Fig:eDiggN}
\begin{tikzpicture}[-]
\scriptsize
\begin{axis}[
ybar=0pt,
enlarge x limits=0.12,
bar width=0.1cm,
legend cell align=left,
legend image code/.code={\draw[#1] (0cm,-0.1cm) rectangle (0.15cm,0.1cm);}, 
legend style={area legend, at={(1,1.02), font=\tiny},anchor=north east,legend columns=1, draw=none, fill=none},
ylabel style={text width=4.2cm, align=center, yshift=-0.7cm},
ylabel=Max. revenue\\(the ratio to NoSocial),
symbolic x coords={0.05, 0.1, 0.15, 0.2, 0.25, 0.3},
xlabel style={align=center},
xlabel={$\frac{n}{|V|}$},
xtick=data,
x tick style={opacity=0},
ymax=4,
ymin=1,
x=0.5cm,
height=4.2cm
]
\addplot
[pattern=crosshatch dots,pattern color=blue]
 coordinates {
(0.05, 3.3126294)
(0.1, 2.826933936)
(0.15, 2.66075167)
(0.2, 2.605621432)
(0.25, 2.366710584)
(0.3, 2.143244035)
};
\addplot
[pattern=crosshatch, pattern color=orange]
 coordinates {
(0.05, 2.756446452)
(0.1, 2.280538302)
(0.15, 2.046123683)
(0.2, 1.956960913)
(0.25, 1.708915239)
(0.3, 1.575483092)
};
\addplot
[pattern=horizontal lines, pattern color=gray]
 coordinates {
(0.05, 1.718784114)
(0.1, 1.596437982)
(0.15, 1.595746519)
(0.2, 1.65410628)
(0.25, 1.53589372)
(0.3, 1.429765042)
};
\legend{PRUB+IF, PRUB+SW, PRUB+R}
\end{axis}
\end{tikzpicture}
}%
\subfigure[][\textit{digg}(M)]{
\centering
\label{Fig:eDiggM}
\begin{tikzpicture}[-]
\scriptsize
\begin{axis}[
ybar=0pt,
enlarge x limits=0.12,
bar width=0.1cm,
legend cell align=left,
legend image code/.code={\draw[#1] (0cm,-0.1cm) rectangle (0.15cm,0.1cm);}, 
legend style={area legend, at={(1,1.02), font=\tiny},anchor=north east,legend columns=1, draw=none, fill=none},
ylabel style={text width=4.2cm, align=center, yshift=-0.7cm},
ylabel=Max. revenue\\(the ratio to NoSocial),
symbolic x coords={0.05, 0.1, 0.15, 0.2, 0.25, 0.3},
xlabel style={align=center},
xlabel={$\frac{n}{|V|}$},
xtick=data,
x tick style={opacity=0},
ymax=4,
ymin=1,
x=0.5cm,
height=4.2cm
]
\addplot
[pattern=crosshatch dots,pattern color=blue]
 coordinates {
(0.05, 2.640901771)
(0.1, 2.278985507)
(0.15, 2.263394818)
(0.2, 2.028573781)
(0.25, 1.848666008)
(0.3, 1.690135046)
};
\addplot
[pattern=crosshatch, pattern color=orange]
 coordinates {
(0.05, 2.243741765)
(0.1, 1.861367296)
(0.15, 1.815985946)
(0.2, 1.592926548)
(0.25, 1.4578722)
(0.3, 1.311923584)
};
\addplot
[pattern=horizontal lines, pattern color=gray]
 coordinates {
(0.05, 1.527477675)
(0.1, 1.44948031)
(0.15, 1.507743193)
(0.2, 1.390577652)
(0.25, 1.275802042)
(0.3, 1.173671498)
};
\legend{PRUB+IF, PRUB+SW, PRUB+R}
\end{axis}
\end{tikzpicture}
}%
\subfigure[][\textit{facebook}(N)]{
\centering
\label{Fig:eFbN}
\begin{tikzpicture}[-]
\scriptsize
\begin{axis}[
ybar=0pt,
enlarge x limits=0.12,
bar width=0.1cm,
legend cell align=left,
legend image code/.code={\draw[#1] (0cm,-0.1cm) rectangle (0.15cm,0.1cm);}, 
legend style={area legend, at={(1,1.02), font=\tiny},anchor=north east,legend columns=1, draw=none, fill=none},
ylabel style={text width=4.2cm, align=center, yshift=-0.7cm},
ylabel=Max. revenue\\(the ratio to NoSocial),
symbolic x coords={0.05, 0.1, 0.15, 0.2, 0.25, 0.3},
xlabel style={align=center},
xlabel={$\frac{n}{|V|}$},
xtick=data,
x tick style={opacity=0},
ymax=8.5,
ymin=1,
x=0.5cm,
height=4.2cm
]
\addplot
[pattern=crosshatch dots,pattern color=blue]
 coordinates {
(0.05,8.24279476)
(0.1,5.878192534)
(0.15,4.577572406)
(0.2,4.418402096)
(0.25,3.57680957)
(0.3,2.980804773)
};
\addplot
[pattern=crosshatch, pattern color=orange]
 coordinates {
(0.05,6.009606987)
(0.1,4.163938005)
(0.15,3.183379421)
(0.2,3.053459943)
(0.25,2.535161093)
(0.3,2.161100196)
};
\addplot
[pattern=horizontal lines, pattern color=gray]
 coordinates {
(0.05,2.449534207)
(0.1,1.988226734)
(0.15,1.716428953)
(0.2,1.813870334)
(0.25,1.627713263)
(0.3,1.4843222)
};
\legend{PRUB+IF, PRUB+SW, PRUB+R}
\end{axis}
\end{tikzpicture}
}%
\subfigure[][\textit{facebook}(M)]{
\centering
\label{Fig:eFbM}
\begin{tikzpicture}[-]
\scriptsize
\begin{axis}[
ybar=0pt,
enlarge x limits=0.12,
bar width=0.1cm,
legend cell align=left,
legend image code/.code={\draw[#1] (0cm,-0.1cm) rectangle (0.15cm,0.1cm);}, 
legend style={area legend, at={(1,1.02), font=\tiny},anchor=north east,legend columns=1, draw=none, fill=none},
ylabel style={text width=4.2cm, align=center, yshift=-0.7cm},
ylabel=Max. revenue\\(the ratio to NoSocial),
symbolic x coords={0.05, 0.1, 0.15, 0.2, 0.25, 0.3},
xlabel style={align=center},
xlabel={$\frac{n}{|V|}$},
xtick=data,
x tick style={opacity=0},
ymax=8.5,
ymin=1,
x=0.5cm,
height=4.2cm
]
\addplot
[pattern=crosshatch dots,pattern color=blue]
 coordinates {
(0.05,6.18209607)
(0.1,4.452412137)
(0.15,3.45888517)
(0.2,2.781816197)
(0.25,2.264057452)
(0.3,1.886796915)
};
\addplot
[pattern=crosshatch, pattern color=orange]
 coordinates {
(0.05,4.520960699)
(0.1,3.160609037)
(0.15,2.409310872)
(0.2,1.933939096)
(0.25,1.617905789)
(0.3,1.382367387)
};
\addplot
[pattern=horizontal lines, pattern color=gray]
 coordinates {
(0.05,1.922751092)
(0.1,1.589030779)
(0.15,1.370462815)
(0.2,1.197219494)
(0.25,1.085406662)
(0.3,1.007929491)
};
\legend{PRUB+IF, PRUB+SW, PRUB+R}
\end{axis}
\end{tikzpicture}
}%
\vspace{-4mm}
\caption{The effectiveness comparison.}
\vspace{-5mm}
\label{Fig:effectiveness}
\end{figure*}
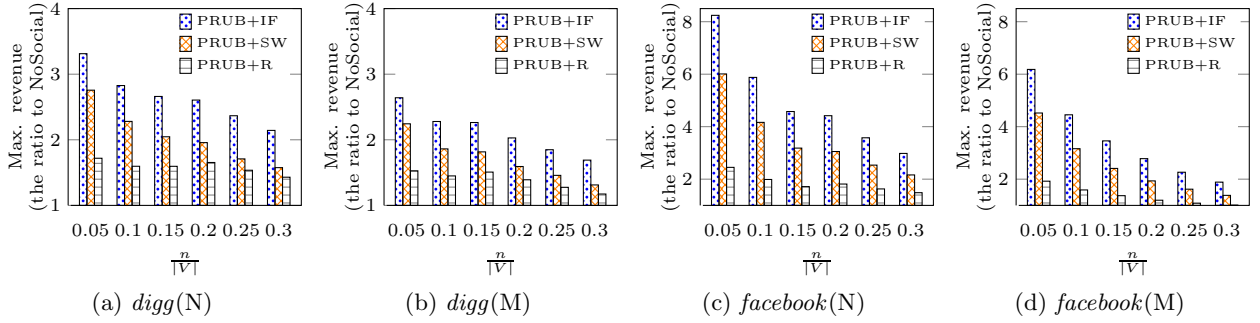

\begin{figure*}[t]
\vspace{-5mm}
\centering
\subfigure[][\textit{digg}(N)]{
\centering
\label{Fig:3DiggN}
\begin{tikzpicture}[-]
\scriptsize
\begin{axis}[
ybar=0pt,
enlarge x limits=0.12,
bar width=0.085cm,
legend cell align=left,
legend image code/.code={\draw[#1] (0cm,-0.1cm) rectangle (0.15cm,0.1cm);}, 
legend style={area legend, at={(1,1.02), font=\tiny},anchor=north east,legend columns=1, draw=none, fill=none},
ylabel style={text width=4.2cm, align=center, yshift=-0.7cm},
ylabel=Max. revenue\\(the ratio to NoSocial),
symbolic x coords={0.05, 0.1, 0.15, 0.2, 0.25, 0.3},
xlabel style={align=center},
xlabel={$\frac{n}{|V|}$},
xtick=data,
x tick style={opacity=0},
ymax=4,
ymin=1,
x=0.5cm,
height=4.2cm
]
\addplot
[pattern=crosshatch dots,pattern color=blue]
 coordinates {
(0.05, 3.3126294)
(0.1, 2.826933936)
(0.15, 2.66075167)
(0.2, 2.605621432)
(0.25, 2.366710584)
(0.3, 2.143244035)
};
\addplot
[pattern=north east lines,pattern color=Emerald]
 coordinates {
(0.05,3.14323357801619)
(0.1,2.66534914361001)
(0.15,2.46475924633636)
(0.2,2.37252964426877)
(0.25,2.10579710144928)
(0.3,1.89975845410628)
};
\addplot
[pattern=my grid,pattern color=Goldenrod]
 coordinates {
(0.05,3.10276679841897)
(0.1,2.52719744024092)
(0.15,2.29262619213771)
(0.2,2.19625603864734)
(0.25,1.92547211242863)
(0.3,1.71629336846728)
};
\addplot
[pattern=north west lines,pattern color=Fuchsia]
 coordinates {
(0.05,2.74308300395257)
(0.1,2.25569358178054)
(0.15,2.04216927524674)
(0.2,1.94115063680281)
(0.25,1.70364514712341)
(0.3,1.57347020933977)
};
\legend{PRUB+IF, PRUB+N, PRUB+F, PRUB+P}
\end{axis}
\end{tikzpicture}
}%
\subfigure[][\textit{digg}(M)]{
\centering
\label{Fig:3DiggM}
\begin{tikzpicture}[-]
\scriptsize
\begin{axis}[
ybar=0pt,
enlarge x limits=0.12,
bar width=0.085cm,
legend cell align=left,
legend image code/.code={\draw[#1] (0cm,-0.1cm) rectangle (0.15cm,0.1cm);}, 
legend style={area legend, at={(1,1.02), font=\tiny},anchor=north east,legend columns=1, draw=none, fill=none},
ylabel style={text width=4.2cm, align=center, yshift=-0.7cm},
ylabel=Max. revenue\\(the ratio to NoSocial),
symbolic x coords={0.05, 0.1, 0.15, 0.2, 0.25, 0.3},
xlabel style={align=center},
xlabel={$\frac{n}{|V|}$},
xtick=data,
x tick style={opacity=0},
ymax=4,
ymin=1,
x=0.5cm,
height=4.2cm
]
\addplot
[pattern=crosshatch dots,pattern color=blue]
 coordinates {
(0.05,2.64090177133655)
(0.1,2.27898550724638)
(0.15,2.26339481774264)
(0.2,2.02857378129117)
(0.25,1.84866600790514)
(0.3,1.69013504611331)

};
\addplot
[pattern=north east lines,pattern color=Emerald]
 coordinates {
(0.05,2.54779680866637)
(0.1,2.14880690967648)
(0.15,2.10391963109354)
(0.2,1.85894268774704)
(0.25,1.66872529644269)
(0.3,1.51162439613527)

};
\addplot
[pattern=my grid,pattern color=Goldenrod]
 coordinates {
(0.05,2.46991655687308)
(0.1,2.06258234519104)
(0.15,1.99824330259113)
(0.2,1.76836297760211)
(0.25,1.54318181818182)
(0.3,1.43220245937637)

};
\addplot
[pattern=north west lines,pattern color=Fuchsia]
 coordinates {
(0.05,2.23042014346362)
(0.1,1.84892402283707)
(0.15,1.81324110671937)
(0.2,1.59177371541502)
(0.25,1.45326086956522)
(0.3,1.31057861220905)

};
\legend{PRUB+IF, PRUB+N, PRUB+F, PRUB+P}
\end{axis}
\end{tikzpicture}
}%
\subfigure[][\textit{facebook}(N)]{
\centering
\label{Fig:3FbN}
\begin{tikzpicture}[-]
\scriptsize
\begin{axis}[
ybar=0pt,
enlarge x limits=0.12,
bar width=0.085cm,
legend cell align=left,
legend image code/.code={\draw[#1] (0cm,-0.1cm) rectangle (0.15cm,0.1cm);}, 
legend style={area legend, at={(1,1.02), font=\tiny},anchor=north east,legend columns=1, draw=none, fill=none},
ylabel style={text width=4.2cm, align=center, yshift=-0.7cm},
ylabel=Max. revenue\\(the ratio to NoSocial),
symbolic x coords={0.05, 0.1, 0.15, 0.2, 0.25, 0.3},
xlabel style={align=center},
xlabel={$\frac{n}{|V|}$},
xtick=data,
x tick style={opacity=0},
ymax=8.5,
ymin=1,
x=0.5cm,
height=4.2cm
]
\addplot
[pattern=crosshatch dots,pattern color=blue]
 coordinates {
(0.05,8.24279476)
(0.1,5.878192534)
(0.15,4.577572406)
(0.2,4.418402096)
(0.25,3.57680957)
(0.3,2.980804773)
};
\addplot
[pattern=north east lines,pattern color=Emerald]
 coordinates {
(0.05,7.04148471615721)
(0.1,5.01244269810085)
(0.15,3.88703730655412)
(0.2,3.75158262388125)
(0.25,3.12009080590238)
(0.3,2.63488321327221)

};
\addplot
[pattern=my grid,pattern color=Goldenrod]
 coordinates {
(0.05,6.98689956331878)
(0.1,4.75554827912392)
(0.15,3.60968806093242)
(0.2,3.43662955686531)
(0.25,2.82986117174539)
(0.3,2.40459870479517)

};
\addplot
[pattern=north west lines,pattern color=Fuchsia]
 coordinates {
(0.05,6.07132459970888)
(0.1,4.18904169395329)
(0.15,3.18672682287877)
(0.2,3.06967910936477)
(0.25,2.56166943158998)
(0.3,2.17699192316088)

};
\legend{PRUB+IF, PRUB+N, PRUB+F, PRUB+P}
\end{axis}
\end{tikzpicture}
}%
\subfigure[][\textit{facebook}(M)]{
\centering
\label{Fig:3FbM}
\begin{tikzpicture}[-]
\scriptsize
\begin{axis}[
ybar=0pt,
enlarge x limits=0.12,
bar width=0.085cm,
legend cell align=left,
legend image code/.code={\draw[#1] (0cm,-0.1cm) rectangle (0.15cm,0.1cm);}, 
legend style={area legend, at={(1,1.02), font=\tiny},anchor=north east,legend columns=1, draw=none, fill=none},
ylabel style={text width=4.2cm, align=center, yshift=-0.7cm},
ylabel=Max. revenue\\(the ratio to NoSocial),
symbolic x coords={0.05, 0.1, 0.15, 0.2, 0.25, 0.3},
xlabel style={align=center},
xlabel={$\frac{n}{|V|}$},
xtick=data,
x tick style={opacity=0},
ymax=8.5,
ymin=1,
x=0.5cm,
height=4.2cm
]
\addplot
[pattern=crosshatch dots,pattern color=blue]
 coordinates {
(0.05,6.182096069869)
(0.1,4.45241213708797)
(0.15,3.45888516955319)
(0.2,2.78181619733683)
(0.25,2.26405745219593)
(0.3,1.88679691479299)

};
\addplot
[pattern=north east lines,pattern color=Emerald]
 coordinates {
(0.05,5.31342794759825)
(0.1,3.80833879065706)
(0.15,2.93063236792316)
(0.2,2.37195754202139)
(0.25,1.96789050903693)
(0.3,1.67054136651386)

};
\addplot
[pattern=my grid,pattern color=Goldenrod]
 coordinates {
(0.05,5.24377729257642)
(0.1,3.60347085789129)
(0.15,2.72977004802794)
(0.2,2.16715782580223)
(0.25,1.79736750196455)
(0.3,1.53007894928327)

};
\addplot
[pattern=north west lines,pattern color=Fuchsia]
 coordinates {
(0.05,4.57893013100437)
(0.1,3.18271119842829)
(0.15,2.40931087177994)
(0.2,1.94376227897839)
(0.25,1.61093163363311)
(0.3,1.37435421669213)

};
\legend{PRUB+IF, PRUB+N, PRUB+F, PRUB+P}
\end{axis}
\end{tikzpicture}
}%
\vspace{-4mm}
\caption{The advantage of each part of PRUB+IF.}
\vspace{-5mm}
\label{Fig:advantage}
\end{figure*}
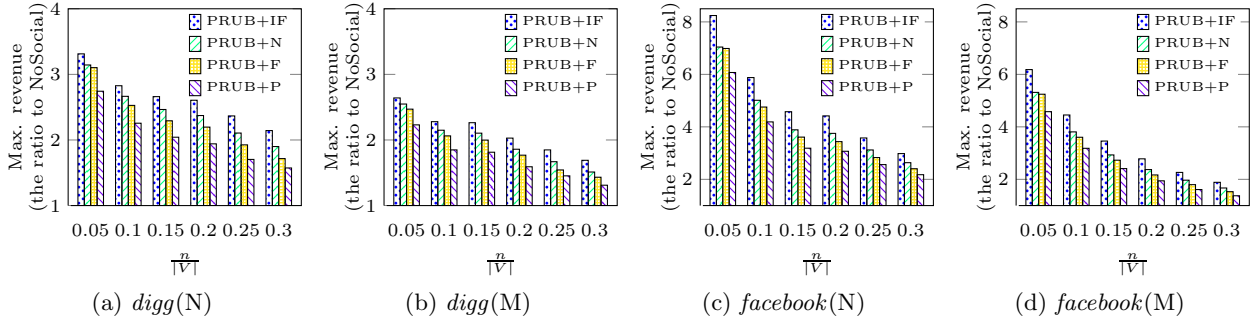

\begin{figure*}[t]
\vspace{-5mm}
\centering
\subfigure[][\textit{digg}(N)]{
\centering
\label{Fig:pDiggN}
\begin{tikzpicture}[-]
\scriptsize
\begin{axis}[
legend cell align=left,
legend style={at={(1.02,1.02), font=\tiny},anchor=north east,legend columns=1, draw=none, fill=none},
ylabel style={text width=4.2cm, align=center, yshift=-0.7cm},
ylabel=Pricing,
symbolic x coords={0.05, 0.1, 0.15, 0.2, 0.25, 0.3},
xlabel style={align=center},
xlabel={$\frac{n}{|V|}$},
xtick=data,
ymax=130,
ymin=0,
x=0.55cm,
height=4.2cm
]
\addplot
[color=blue, mark=triangle, mark options={solid, scale=1.5}]
 coordinates {
(0.05, 40)
(0.1, 34)
(0.15, 34)
(0.2, 34)
(0.25, 34)
(0.3, 28)
};
\addplot
[color=orange, mark=asterisk, mark options={solid, scale=1.5}]
 coordinates {
(0.05, 29)
(0.1, 22)
(0.15, 17)
(0.2, 16)
(0.25, 16)
(0.3, 11)
};
\addplot
[color=gray, mark=square, mark options={solid, scale=1.5}]
 coordinates {
(0.05, 15.7)
(0.1, 11.9)
(0.15, 11.1)
(0.2, 11)
(0.25, 11)
(0.3, 10.3)
};
\addplot
[color=magenta, mark=o, mark options={solid, scale=1.5}]
 coordinates {
(0.05, 7)
(0.1, 7)
(0.15, 7)
(0.2, 6)
(0.25, 6)
(0.3, 6)
};
\legend{PRUB+IF, PRUB+SW, PRUB+R, NoSocial}
\end{axis}
\end{tikzpicture}
}%
\subfigure[][\textit{digg}(M)]{
\centering
\label{Fig:pDiggM}
\begin{tikzpicture}[-]
\scriptsize
\begin{axis}[
legend cell align=left,
legend style={at={(1.02,1.02), font=\tiny},anchor=north east,legend columns=1, draw=none, fill=none},
ylabel style={text width=4.2cm, align=center, yshift=-0.7cm},
ylabel=Pricing,
symbolic x coords={0.05, 0.1, 0.15, 0.2, 0.25, 0.3},
xlabel style={align=center},
xlabel={$\frac{n}{|V|}$},
xtick=data,
ymax=130,
ymin=0,
x=0.55cm,
height=4.2cm
]
\addplot
[color=blue, mark=triangle, mark options={solid, scale=1.5}]
 coordinates {
(0.05, 44)
(0.1, 37)
(0.15, 31)
(0.2, 26)
(0.25, 31)
(0.3, 25)
};
\addplot
[color=orange, mark=asterisk, mark options={solid, scale=1.5}]
 coordinates {
(0.05, 26)
(0.1, 20)
(0.15, 20)
(0.2, 14)
(0.25, 14)
(0.3, 14)
};
\addplot
[color=gray, mark=square, mark options={solid, scale=1.5}]
 coordinates {
(0.05, 14)
(0.1, 13.9)
(0.15, 13.8)
(0.2, 13.8)
(0.25, 13.3)
(0.3, 13.1)
};
\addplot
[color=magenta, mark=o, mark options={solid, scale=1.5}]
 coordinates {
(0.05, 9)
(0.1, 9)
(0.15, 8)
(0.2, 8)
(0.25, 8)
(0.3, 8)
};
\legend{PRUB+IF, PRUB+SW, PRUB+R, NoSocial}
\end{axis}
\end{tikzpicture}
}%
\subfigure[][\textit{facebook}(N)]{
\centering
\label{Fig:pFbN}
\begin{tikzpicture}[-]
\scriptsize
\begin{axis}[
legend cell align=left,
legend style={at={(1.02,1.02), font=\tiny},anchor=north east,legend columns=1, draw=none, fill=none},
ylabel style={text width=4.2cm, align=center, yshift=-0.7cm},
ylabel=Pricing,
symbolic x coords={0.05, 0.1, 0.15, 0.2, 0.25, 0.3},
xlabel style={align=center},
xlabel={$\frac{n}{|V|}$},
xtick=data,
ymax=130,
ymin=0,
x=0.55cm,
height=4.2cm
]
\addplot
[color=blue, mark=triangle, mark options={solid, scale=1.5}]
 coordinates {
(0.05,99)
(0.1,99)
(0.15,69)
(0.2,53)
(0.25,47)
(0.3,47)
};
\addplot
[color=orange, mark=asterisk, mark options={solid, scale=1.5}]
 coordinates {
(0.05,84)
(0.1,70)
(0.15,46)
(0.2,43)
(0.25,46)
(0.3,44)
};
\addplot
[color=gray, mark=square, mark options={solid, scale=1.5}]
 coordinates {
(0.05,18.7)
(0.1,18.3)
(0.15,17.1)
(0.2,16.8)
(0.25,17.4)
(0.3,18)
};
\addplot
[color=magenta, mark=o, mark options={solid, scale=1.5}]
 coordinates {
(0.05,6)
(0.1,6)
(0.15,6)
(0.2,5)
(0.25,5)
(0.3,5)
};
\legend{PRUB+IF, PRUB+SW, PRUB+R, NoSocial}
\end{axis}
\end{tikzpicture}
}%
\subfigure[][\textit{facebook}(M)]{
\centering
\label{Fig:pFbM}
\begin{tikzpicture}[-]
\scriptsize
\begin{axis}[
legend cell align=left,
legend style={at={(1.02,1.02), font=\tiny},anchor=north east,legend columns=1, draw=none, fill=none},
ylabel style={text width=4.2cm, align=center, yshift=-0.7cm},
ylabel=Pricing,
symbolic x coords={0.05, 0.1, 0.15, 0.2, 0.25, 0.3},
xlabel style={align=center},
xlabel={$\frac{n}{|V|}$},
xtick=data,
ymax=130,
ymin=0,
x=0.55cm,
height=4.2cm
]
\addplot
[color=blue, mark=triangle, mark options={solid, scale=1.5}]
 coordinates {
(0.05,117)
(0.1,76)
(0.15,68)
(0.2,56)
(0.25,38)
(0.3,38)
};
\addplot
[color=orange, mark=asterisk, mark options={solid, scale=1.5}]
 coordinates {
(0.05,84)
(0.1,66)
(0.15,45)
(0.2,45)
(0.25,39)
(0.3,39)
};
\addplot
[color=gray, mark=square, mark options={solid, scale=1.5}]
 coordinates {
(0.05,20.7)
(0.1,16.9)
(0.15,16.7)
(0.2,15.3)
(0.25,11.9)
(0.3,9.2)
};
\addplot
[color=magenta, mark=o, mark options={solid, scale=1.5}]
 coordinates {
(0.05,8)
(0.1,8)
(0.15,8)
(0.2,8)
(0.25,8)
(0.3,8)
};
\legend{PRUB+IF, PRUB+SW, PRUB+R, NoSocial}
\end{axis}
\end{tikzpicture}
}%
\vspace{-4mm}
\caption{The pricing comparison.}
\vspace{-5mm}
\label{Fig:pricing}
\end{figure*}
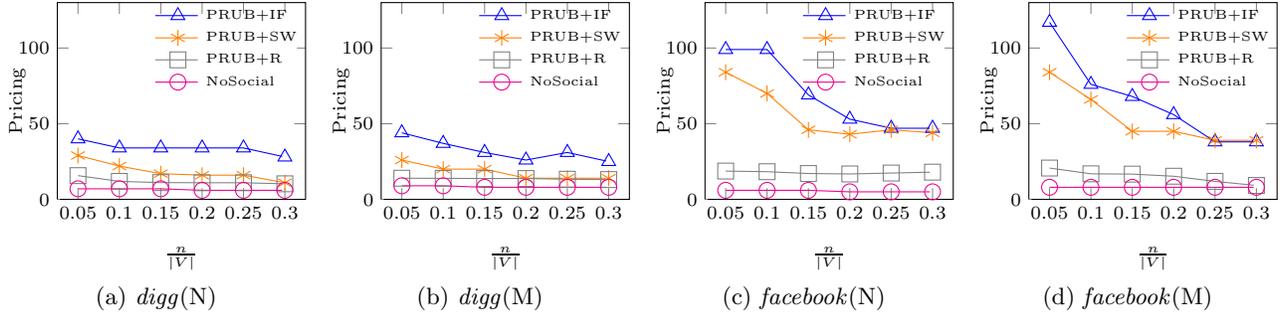

\textbf{Effectiveness.}
Here we use the maximum revenue obtained without consideration of social influences (NoSocial), i.e., the adoption of an individual only depends on his/her inherent valuation, as the baseline, and show the effectiveness of the heuristics in ratios to the baseline. 
Figure~\ref{Fig:effectiveness} reports the effectiveness
with respect to different ratios of the quantity of commodities to the number of population, i.e., $\frac{n}{|V|}$. 
For these results, we have two
observations as follows.
First, on \textit{digg} and \textit{facebook} with both Normal and M-shape valuation distributions, PRUB+IF is the best, and PRUB+SW outperforms PRUB+R only in
the cases of small $\frac{n}{|V|}$. 
The reason is that, to estimate the advantage of selecting an individual in the promotion,  
PRUB+IF introduces the pricing-sensitive importance,
which considers the normalized effects of the influence propagated through the social network, while  
PRUB+SW only takes the sum of weights on one's direct out-neighbors. RRUB+R picks the seeds blindly. 
Second, as $\frac{n}{|V|}$ decreases, the performance of
PRUB+IF becomes even more noticeable (compared to
PRUB+SW and PRUB+R). This is because the number of potential buyers
reduces when the pricing of the commodity increases due to scarcity.
The less the number of potential buyers is, the more
challenging the seed selection is.
Hence, the results demonstrate that the strategy of PRUB+IF for the seed
selection outperforms those of PRUB+SW and PRUB+R.

In order to demonstrate the advantage of each part of PRUB+IF, Figure~\ref{Fig:advantage} shows the effectiveness of PRUB+IF, PRUB+N, PRUB+F, and PRUB+P with respect to different ratios of the quantity of commodities to the number of population, i.e., $\frac{n}{|V|}$. Accordingly, the three parts of PRUB+IF are more or less all equally important.

\textbf{Pricing.}
In addition to the demonstration of effectiveness by maximum
revenue, whether or not the suggested pricing complies with the common sense is
also an interesting question. For answering this question, we
plotted the trend of the pricing suggested by PRUB+IF, PRUB+SW, PRUB+R, and NoSocial with respect to different ratios of the quantity of commodities to the number of population, i.e., $\frac{n}{|V|}$, in Figure~\ref{Fig:pricing}. Accordingly, 
PRUB+IF suggests higher pricing in most cases than the other approaches to gain higher maximum revenue (as shown in Figure~\ref{Fig:effectiveness}) in all cases. This implies that PRUB+IF performs the most elegant strategy for seed selection. 
The smart utilization of word-of-mouth effects can help the promotion of commodities in higher pricing, which complies with the common sense. 
Furthermore, for all methods, 
there is a upward trend in the pricing as $\frac{n}{|V|}$ decreases. 
The less the quantity is, the more precious the commodity could be, which also complies with the common sense. 
There are some exceptions to this pricing trend: PRUB+IF at 0.25 on \textit{digg}(M), and both PRUB+SW and PRUB+R at 0.25 or 0.3 on \textit{facebook}(N).
The reason is that a quantity of supply is used as freebies. This may allow higher pricing in some cases by leaving less quantities for sale. 
The tiny example in Figure~\ref{Fig:exP} explains such a case in more details. ($F(x)=x$ is used as the concave influence function.)
For the case of the quantity of the concert tickets as 2, the maximum revenue is $R(2,\$3,\emptyset)=\$6$ (i.e., $b$ and $c$ adopt the
concert tickets).
Now consider that the quantity of the concert tickets is increased to 3. The maximum revenue
is $R(3,\$7,\{a,c\})=\$7$ (i.e., $b$ adopts the concert ticket under the impact of $a$ and $c$). The
pricing is higher even though the quantity increases. This is because
2 out of 3 concert tickets are used as free tickets and only 1 concert ticket is left
for sale.
Hence, according to these results, PRUB+IF is capable of suggesting elegant marketing strategies.

\textbf{Efficiency.}
Now we compare the efficiency of the three heuristics.
Figure~\ref{Fig:runtime} shows the runtime with respect to different ratios of the quantity of commodities to the number of population, i.e., $\frac{n}{|V|}$, on \textit{digg}(N),
\textit{digg}(M), \textit{facebook}(N), and \textit{facebook}(M).
PRUB+IF spends more execution time than PRUB+SW and PRUB+R in all cases, since PRUB+IF considers the influence propagation in the derivation of pricing-sensitive importance for seed selection. PRUB+R generally spends a little bit more time than PRUB+SW due to the blind seed selection (that makes PRUB+R need to search more prices for finding maximum revenue).  
Nevertheless, the runtime of all the three heuristics is not sensitive to $\frac{n}{|V|}$. 
For PRUB+IF, other factors such as the local structure of potential buyers and the influence cascades will also affect the runtime.

\begin{figure}[t]
\vspace{1mm}
\centering
\begin{tikzpicture}
    \node[main node] (a) {$a(\$1)$};
    \node[main node] (b) [right=1cm of a] {$b(\$3)$};
    \node[main node] (c) [right=1cm of b] {$c(\$3)$};

    \path
    	(a) [->] edge node {2} (b)
        (c) [->] edge node {2} (b);
\end{tikzpicture}
\vspace{-3mm}
\caption{An example for the exception to the pricing trend.}
\label{Fig:exP}
\end{figure}
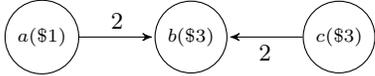

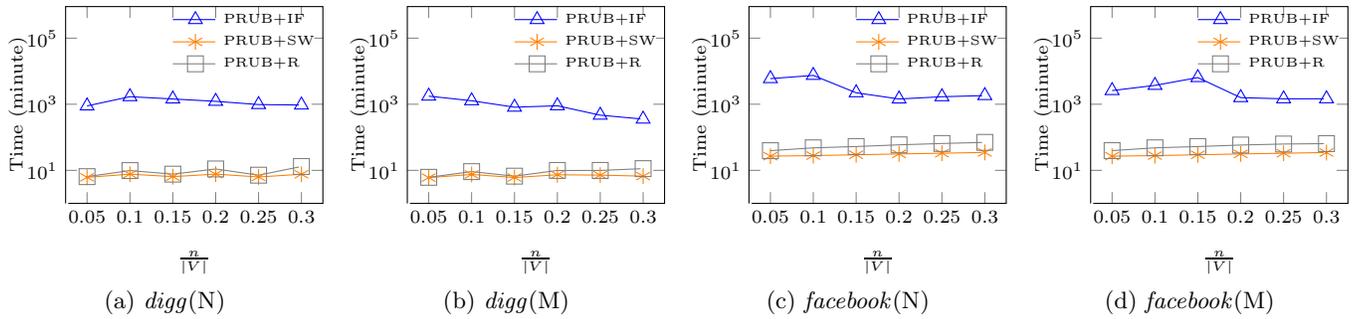
\begin{figure*}[t]
\vspace{-5mm}
\centering
\subfigure[][\textit{digg}(N)]{
\centering
\label{Fig:tDiggN}
\begin{tikzpicture}[-]
\scriptsize
\begin{semilogyaxis}[
legend cell align=left,
legend style={at={(1.02,1.02), font=\tiny},anchor=north east,legend columns=1, draw=none, fill=none},
ylabel style={text width=4.2cm, align=center, yshift=-0.6cm},
ylabel=Time (minute),
symbolic x coords={0.05, 0.1, 0.15, 0.2, 0.25, 0.3},
xlabel style={align=center},
xlabel={$\frac{n}{|V|}$},
xtick=data,
ymax=900000,
ymin=1,
x=0.57cm,
height=4.2cm
]
\addplot
[color=blue, mark=triangle, mark options={solid, scale=1.5}]
 coordinates {
(0.05, 888.9666667)
(0.1, 1691.983333)
(0.15, 1442.05)
(0.2, 1225.183333)
(0.25, 977.2333333)
(0.3, 946.85)
};
\addplot
[color=orange, mark=asterisk, mark options={solid, scale=1.5}]
 coordinates {
(0.05, 6.15)
(0.1, 7.666666667)
(0.15, 6.45)
(0.2, 7.7)
(0.25, 6.3)
(0.3, 7.616666667)
};
\addplot
[color=gray, mark=square, mark options={solid, scale=1.5}]
 coordinates {
(0.05, 6.453333333)
(0.1, 9.848333333)
(0.15, 7.68)
(0.2, 11.015)
(0.25, 7.141666667)
(0.3, 13.07666667)
};
\legend{PRUB+IF, PRUB+SW, PRUB+R}
\end{semilogyaxis}
\end{tikzpicture}
}%
\subfigure[][\textit{digg}(M)]{
\centering
\label{Fig:tDiggM}
\begin{tikzpicture}[-]
\scriptsize
\begin{semilogyaxis}[
legend cell align=left,
legend style={at={(1.02,1.02), font=\tiny},anchor=north east,legend columns=1, draw=none, fill=none},
ylabel style={text width=4.2cm, align=center, yshift=-0.6cm},
ylabel=Time (minute),
symbolic x coords={0.05, 0.1, 0.15, 0.2, 0.25, 0.3},
xlabel style={align=center},
xlabel={$\frac{n}{|V|}$},
xtick=data,
ymax=900000,
ymin=1,
x=0.57cm,
height=4.2cm
]
\addplot
[color=blue, mark=triangle, mark options={solid, scale=1.5}]
 coordinates {
(0.05, 1758.9)
(0.1, 1274.466667)
(0.15, 813.8166667)
(0.2, 897.15)
(0.25, 466.6833333)
(0.3, 354.9333333)
};
\addplot
[color=orange, mark=asterisk, mark options={solid, scale=1.5}]
 coordinates {
(0.05, 6.033333333)
(0.1, 7.5)
(0.15, 6.066666667)
(0.2, 7.35)
(0.25, 7.166666667)
(0.3, 6.683333333)
};
\addplot
[color=gray, mark=square, mark options={solid, scale=1.5}]
 coordinates {
(0.05, 6.158333333)
(0.1, 9.143333333)
(0.15, 6.671666667)
(0.2, 9.92)
(0.25, 9.931666667)
(0.3, 11.27)
};
\legend{PRUB+IF, PRUB+SW, PRUB+R}
\end{semilogyaxis}
\end{tikzpicture}
}%
\subfigure[][\textit{facebook}(N)]{
\centering
\label{Fig:tFbN}
\begin{tikzpicture}[-]
\scriptsize
\begin{semilogyaxis}[
legend cell align=left,
legend style={at={(1.02,1.02), font=\tiny},anchor=north east,legend columns=1, draw=none, fill=none},
ylabel style={text width=4.2cm, align=center, yshift=-0.6cm},
ylabel=Time (minute),
symbolic x coords={0.05, 0.1, 0.15, 0.2, 0.25, 0.3},
xlabel style={align=center},
xlabel={$\frac{n}{|V|}$},
xtick=data,
ymax=900000,
ymin=1,
x=0.57cm,
height=4.2cm
]
\addplot
[color=blue, mark=triangle, mark options={solid, scale=1.5}]
 coordinates {
(0.05,5949.983333)
(0.1,7454.3)
(0.15,2222.05)
(0.2,1440.9)
(0.25,1686.45)
(0.3,1826.216667)
};
\addplot
[color=orange, mark=asterisk, mark options={solid, scale=1.5}]
 coordinates {
(0.05,26.83333333)
(0.1,27.93333333)
(0.15,29.55)
(0.2,31.35)
(0.25,32.93333333)
(0.3,34.86666667)
};
\addplot
[color=gray, mark=square, mark options={solid, scale=1.5}]
 coordinates {
(0.05,39.36833333)
(0.1,47.925)
(0.15,52.51666667)
(0.2,59.10666667)
(0.25,64.87)
(0.3,70.15)
};
\legend{PRUB+IF, PRUB+SW, PRUB+R}
\end{semilogyaxis}
\end{tikzpicture}
}%
\subfigure[][\textit{facebook}(M)]{
\centering
\label{Fig:tFbM}
\begin{tikzpicture}[-]
\scriptsize
\begin{semilogyaxis}[
legend cell align=left,
legend style={at={(1.02, 1.02), font=\tiny},anchor=north east,legend columns=1, draw=none, fill=none},
ylabel style={text width=4.2cm, align=center, yshift=-0.6cm},
ylabel=Time (minute),
symbolic x coords={0.05, 0.1, 0.15, 0.2, 0.25, 0.3},
xlabel style={align=center},
xlabel={$\frac{n}{|V|}$},
xtick=data,
ymax=900000,
ymin=1,
x=0.57cm,
height=4.2cm
]
\addplot
[color=blue, mark=triangle, mark options={solid, scale=1.5}]
 coordinates {
(0.05,2587.55)
(0.1,3719.333333)
(0.15,6395.5)
(0.2,1584.016667)
(0.25,1458.666667)
(0.3,1461.65)
};
\addplot
[color=orange, mark=asterisk, mark options={solid, scale=1.5}]
 coordinates {
(0.05,26.85)
(0.1,27.93333333)
(0.15,29.55)
(0.2,31.35)
(0.25,32.91666667)
(0.3,34.58333333)
};
\addplot
[color=gray, mark=square, mark options={solid, scale=1.5}]
 coordinates {
(0.05,39.52833333)
(0.1,48.00833333)
(0.15,52.405)
(0.2,58.29166667)
(0.25,62.54)
(0.3,64.29666667)
};
\legend{PRUB+IF, PRUB+SW, PRUB+R}
\end{semilogyaxis}
\end{tikzpicture}
}%
\vspace{-4mm}
\caption{The runtime comparison.}
\vspace{-5mm}
\label{Fig:runtime}
\end{figure*}

\section{Related works}
\textbf{Influence Maximization.} The influence maximization problem
aims at identifying a group of seeds so that the number of the
active people is the largest \cite{Chen2010, Domingos2001, Kempe2003, Singer2012}. To the best of our
knowledge, the first study could be traced to the work \cite{Domingos2001},
where Domingos et al.\ studied the influence maximization problem as an
algorithmic problem. Then, Kempe et al.\ \cite{Kempe2003} formally
modeled the problem as a discrete optimization, which is to identify
$k$ seeds for maximizing the influence spread over the social network.
Meanwhile, Kempe et al.\ also showed that the influence maximization
problem is NP-hard under both the Linear Threshold (LT) and the Independent
Cascade (IC) models, and proposed a greedy strategy with approximation
guarantees.
Specifically, in the LT model, the influences on an individual are
cumulative, and each individual has his/her own activation threshold.
An individual becomes active if the sum of weights from his/her active
neighbors is larger than or equal to his/her activation threshold.
Then, this active individual will also propagate influences to his/her
inactive neighbors. In the IC model,
each influence an individual receives is independent. For each
influence, there is one and only one chance to activate the individual
with a probability.
If the individual becomes active, he/she futher spreads his/her
influence in the same mannar.
In the work \cite{Chen2010}, Chen et al.\ showed that, given seeds, computing
the influence spread under the LT model is \#P-hard, and provided a
scalable heuristic based on the fast computation for directed acyclic
graphs (DAGs).
Singer \cite{Singer2012} concerned the inherent costs of making
individuals as seeds and followed the incentive compatible mechanism
to extract the true information about individuals' costs for
identifying and rewarding the seeds.
As mentioned in Section~\ref{Sec:Intro}, all these influence
maximization problems do not take the monetary effects on people's purchase decisions into account,
and thus differ from our problem.

\textbf{Revenue Maximization.} The problem of revenue maximization is
to derive the marketing strategy that brings the optimal revenue under
the social influence, and was first addressed by
Hartline et al.\ \cite{Hartline2008}.
As a solution, Hartline et al.\ \cite{Hartline2008} proposed a marketing strategy, named
influence-and-exploit (IE), to tackle the problem in two steps. In the
influence step, IE identifies a group
of customers as seeds. After that, in the exploit step, IE visits the remaining
customers in a random order and offers each of them the optimal
(myopic) pricing of the product in order to maximize the revenue.
The decision whether or not an individual purchases the product
depends on the pricing offered and the individual's valuation, which
takes the social influences into consideration. Hartline et al.\
\cite{Hartline2008} also discussed the inferences of the valuations
from social influences by Uniform Additive Model, Symmetric Model,
and Concave Graph Model. Later, Fotakis et al.\ \cite{Fotakis2012}
proposed a polynomial-time approach to approximate the maximum
revenue in the Uniform Additive Model.
Another work \cite{Lu2012} distinguishes
between activation and adoption. When an individual is active by the
social influence, he/she will then determine whether or not to adopt
the product with respect to the
pricing and his/her valuation.
The marketing strategy of all these
works is to determine customized pricing for different customers.

Mirrokni et al.\ \cite{Mirrokni2012} argued that offering the
product at fixed pricing is more reasonable since
offering customized pricing is hard to implement in the real world.
Hence, their objective is to discover
a seed group with the fixed pricing to maximize the revenue.
Mirrokni et al.\ \cite{Mirrokni2012} adopted the Concave Graph Model proposed in the work
\cite{Hartline2008} to incorporate the social influence into the
valuation, and presented a (1/2)-approximation based on the randomized linear time approximation algorithm for the Unconstrained Submodular Maximization problem.

However, all the works above assumed that there is an unlimited supply of commodities, whereas we consider the scenario with limited supply such as limited edition commodities and concert tickets.
As shown in Sections~\ref{Sec:Intro} and \ref{Sec:exp1},
it is unapplicable for these works to determine the proper marketing
strategies.

\vspace{-2mm}
\section{Conclusion}
In this work, we addressed the revenue maximization problem with a quantity constraint on a monetizing social network. 
To maximize the revenue, the marketing strategy is to determine the pricing of the commodity and find a seed group of individuals as the initial customers.
Hence, we provided Algorithm PRUB to derive the pricing and the seed group for the optimal revenue.
For better efficiency, we further proposed Algorithm PRUB+IF as a heuristic algorithm based on the framework of PRUB for a feasible solution.
The experiments on real datasets with different valuation distributions showed the effectiveness of PRUB and PRUB+IF.


\vspace{-2mm}
\begin{thebibliography}{10}
\bibitem{Arthur2009} D. Arthur, R. Motwani, A. Sharma, and Y. Xu. Pricing Strategies for Viral Marketing on Social Networks. In \textit{WINE}, 2009.
\bibitem{Asch1951} S. E. Asch. Effects of Group Pressure on the Modification and Distortion of Judgments. In H. Guetzkow (Ed.), \textit{Groups, leadership and men}(pp. 177-190). Pittsburgh, PA:Carnegie Press, 1951.
\bibitem{Chen2010} W. Chen, Y. Yuan, and L. Zhang. Scalable Influence Maximization in Social Networks under the Linear Threshold Model. In {\it{IEEE ICDM}}, 2010.
\bibitem{Choudhury2009} M. D. Choudhury, H. Sundaram, A. John, D. D. Seligmann. Social Synchrony: Predicting Mimicry of User Actions in Online Social Media. In {\it{IEEE CSE}}, 2009.
\bibitem{Clauset2009} A. Clauset, C. R. Shalizi, and M. E. J. Newman. Power-law Distributions in Empirical Data. In \textit{SIAM review} 51.4, 2009.
\bibitem{Coleman1964} J. S. Coleman. Introduction to Mathematical Sociology. \textit{London Free Press Glencoe}, 1964. 
\bibitem{Domingos2001} P. Domingos and M. Richardson. Mining the Network Value of Customers. In {\it{ACM SIGKDD}}, 2001.
\bibitem{Easley2010} D. Easley and J. Kleinberg. Information Cascades. In Networks, Crowds, and Markets: Reasoning about a Highly Connected World (chap. 16, pp. 483-508). Cambridge University Press, 2010.
\bibitem{Fotakis2012} D. Fotakis and P. Siminelakis. On the Efficiency of Influence-and-Exploit Strategies for Revenue Maximization under Positive Externalities. In \textit{WINE}, 2012. 
\bibitem{Fromlet2012} H. Fromlet. Predictability of Financial Crises: Lessons from Sweden for Other Countries. \textit{Business Economics} 47.4, 2012.
\bibitem{Hartline2008} J. Hartline, V. S. Mirrokni, and M. Sundararajan. Optimal Marketing Strategies over Social Networks. In {\it{ACM WWW}}, 2008.
\bibitem{Hu2012} B. Hu, M. Jamali, and M. Ester. Learning the Strength of the Factors Influencing User Behavior in Online Social Networks. In {\it{IEEE ASONAM}}, 2012.
\bibitem{Jung2005} S. Y. Jung, J.-H. Hong, and T.-S. Kim. A Statistical Model for User Preference. In \textit{IEEE TKDE}, 2005.
\bibitem{Katz1985} M. L. Katz and C. Shapiro. Network Externalities, Competition, and Compatibility. \textit{American Economic Review} 75: 424-440, 1985.
\bibitem{Kempe2003} D. Kempe, J. Kleinberg, and \'{E}. Tardos. Maximizing the Spread of Influence through a Social Network. In {\it{ACM SIGKDD}}, 2003.
\bibitem{Lu2012} W. Lu and L. V. Lakshmanan. Profit Maximization over Social Networks. In \textit{IEEE ICDM}, 2012.
\bibitem{Mirrokni2012} V. S. Mirrokni, S. Roch, and M. Sundararajan. On Fixed-Price Marketing for Goods with Positive Network Externalities. In {\it{WINE}}, 2012.
\bibitem{Ohmaa2006} K. Ohmaa. The Impact of Rising Lower-Middle Class Population in Japan. Tokyo: Kodan-sha Publishing Company, 2006.
\bibitem{Singer2012} Y. Singer. How to Win Friends and Influence People, Truthfully: Influence Maximization Mechanisms for Social Networks. In \textit{ACM WSDM}, 2012.
\bibitem{Sundararajan2007} A. Sundararajan. Local Network Effects and Complex Network Structure. \textit{The BE Journal of Theoretical Economics} 7(1), 2007.
\bibitem{Worchel1975} S. Worchel, J. Lee, and A. Adewole. Effects of Supply and Demand on Ratings of Object Value. \textit{Journal of Personality and Social Psychology} 32.5 (1975): 906.
\bibitem{Viswanath2009} B. Viswanath, A. Mislove, M. Cha, and K. P. Gummadi. On the Evolution of User Interaction in Facebook. In  \textit{ACM WOSN}, 2009.



\end{thebibliography}
\end{document}